\documentclass[journal,comsoc]{IEEEtran}
\usepackage{blindtext, graphicx}
\usepackage{amssymb}
\usepackage{amsmath}
\usepackage{graphicx}
\usepackage{cite}
\usepackage{citesort}
\usepackage{multicol}
\usepackage{amsfonts}
\usepackage{color}
\usepackage{subfigure}
\usepackage{graphicx,epstopdf}
\usepackage{epsfig}	
\usepackage{cite,graphicx,amsmath,amssymb}
\usepackage{comment}
\usepackage{amsmath}
\allowdisplaybreaks[4]
\usepackage{lipsum}
\usepackage{stfloats}

\newtheorem{proposition}{{Proposition}}

\newtheorem{lemma}{\textbf{Lemma}}

\newtheorem{corollary}{\textbf{Corollary}}

\newtheorem{proof}{\textbf{Proof}}

\makeatletter
\def\ScaleIfNeeded{%
\ifdim\Gin@nat@width>\linewidth \linewidth \else \Gin@nat@width
\fi } \makeatother

\begin{document}
%

\title{\huge{Wireless Powered Dense Cellular Networks:\\How Many Small Cells Do We Need?}}
\author{Lifeng Wang,~\IEEEmembership{Member,~IEEE,}  Kai-Kit Wong,~\IEEEmembership{Fellow,~IEEE}, \\ Robert W. Heath, Jr.,~\IEEEmembership{Fellow,~IEEE}, and Jinhong Yuan,~\IEEEmembership{Fellow,~IEEE}
\thanks{L. Wang, and K.-K. Wong are with the Department of Electronic and Electrical Engineering, University College London, WC1E 7JE, London, UK (E-mail: $\rm\{lifeng.wang, kai$-$\rm kit.wong\}@ucl.ac.uk$).}
\thanks{Robert W. Heath, Jr. is with the Department of Electrical and Computer Engineering, The University of Texas at Austin, Texas, USA (E-mail: $\rm{rheath}@ece.utexas.edu$).}
\thanks{J. Yuan is with School of Electrical Engineering and Telecommunications, University of New South Wales, Sydney, Australia (E-mail: $\rm j.yuan@unsw.edu.au$).}
}

\maketitle

\begin{abstract}
This paper focuses on wireless powered 5G dense cellular networks, where base station (BS) delivers energy to user equipment (UE) via the microwave radiation in sub-6 GHz or millimeter wave (mmWave) frequency, and UE uses the harvested energy for uplink information transmission.  By addressing the impacts of employing different number of antennas and bandwidths at lower and higher frequencies,  we evaluate the amount of harvested energy and throughput in such networks. Based on the derived results, we obtain the required small cell density to achieve an expected level of harvested energy or throughput. Also, we obtain that when the ratio of the number of sub-6 GHz BSs to that of the mmWave BSs is lower than a given threshold,  UE harvests more energy from a mmWave BS than a sub-6 GHz BS. We find  how many mmWave small cells are needed to perform better than the sub-6 GHz small cells from the perspectives of harvested energy and throughput. Our results reveal that the amount of harvested energy from the mmWave tier can be comparable to the sub-6 GHz counterpart in the dense scenarios. For the same tier scale,  mmWave tier can achieve higher throughput. Furthermore, the throughput gap between different mmWave frequencies increases with the mmWave BS density.
\end{abstract}

\begin{IEEEkeywords}
 Millimeter wave (mmWave), dense small cells, wireless power transfer, throughput.
\end{IEEEkeywords}

\section{Introduction}
Wireless power transfer (WPT) is an appealing approach to prolong the lifetime of user equipment (UE), when compared to the traditional energy harvesting sources such as solar and wind that highly depend upon the conditions of the environments. However, the implementation of WPT in conventional cellular networks may be challenging, due to the fact that it cannot support long-distance WPT in the absence of directed power beams, and small cells are not densely deployed~\cite{Kaibin_Huang2015}.  In the fifth generation (5G) cellular networks, technologies such as  millimeter wave (mmWave),  massive multiple-input multiple-output (MIMO) and ultra-dense small cells in the sub-6 GHz and mmWave frequencies will be adopted~\cite{F_Boccardi_5G,Jeffrey_5G,dantong-survey}, which make next-generation networks more suitable for WPT, due to at least the following two key factors:
\begin{itemize}
  \item The very sharp signal beams in large-scale antenna systems such as massive MIMO and mmWave bring large antenna array gains, enabling WPT over long distances.
  \item  Dense small cells will be deployed to provide proximity services, which drastically reduce propagation loss for WPT. In 5G ultra-dense networks (UDNs)~\cite{Ge_X_2016}, the distance between a UE and its serving base station (BS) will be much shorter than ever before.
\end{itemize}
Therefore, 5G networks provide a wealth of opportunities for WPT.

In this paper, we study wireless powered dense cellular networks, in which active UE may select a sub-6 GHz or mmWave BS as dedicated RF energy source, and utilizes its harvested energy for uplink information transmission. We provide
a tractable analytical framework to characterize both the energy
harvesting and throughput performance in such networks. This work will answer how many sub-6 GHz/mmWave small cells are needed in order to achieve some target harvested energy and throughput. In particular, we derive the number of mmWave small cells that is required to achieve better performance than the sub-6 GHz counterpart.

\subsection{Prior Work and Motivation}
Early works have studied the potential of wireless energy harvesting in traditional systems. In \cite{Rui_Zhang_2013}, simultaneous wireless information and power transfer (SWIPT) was investigated in a MIMO wireless broadcasting channel, and two RF energy harvesting designs at the receiver were investigated, namely time switching and power splitting. In \cite{CheDZ14},  each single-antenna UE was considered to harvest ambient RF energy from the surrounding single-antenna access points in a wireless powered network, and a spatial throughput maximization problem was formulated. The work of \cite{Flint_2015} then studied wireless energy harvesting in the sensor networks, where many battery-free sensors are powered by a number of ambient RF energy sources. The power allocation problem in a wireless energy harvesting enabled relay network was considered in ~\cite{zhiguo_Ding_2014}, where the energy constrained relay used power splitting for cooperative transmission. In~\cite{S_A_H_WPT_2015},  $K$-tier uplink cellular networks with ambient RF energy harvesting were considered, and the uplink coverage probability was derived. In device-to-device (D2D) underlaying homogeneous cellular networks, \cite{Sakr_AH_2015} investigated wireless energy harvesting enabled D2D transmissions under different spectrum access policies. Most recently in heterogeneous cellular network (HCN) with energy harvesting based D2Ds, \cite{HH_Yang_2016} studied the D2D relaying in D2D communications. However, the prior work~\cite{Rui_Zhang_2013,CheDZ14,Flint_2015,zhiguo_Ding_2014,S_A_H_WPT_2015,Sakr_AH_2015,HH_Yang_2016} only addresses the effects of  current network features on WPT. Therefore, new research on WPT under emerging 5G network architectures is needed.

The rapid development of 5G technologies has encouraged more research on WPT. In \cite{Gang_yang2015}, the optimal power transfer beamforming was asymptotically derived by considering large number of antennas in a single massive MIMO cell, and the optimal solution for maximizing the throughput under the user fairness criterion was also asymptotically obtained. The work of \cite{Kai_kit_Wong_2015_Mag} provided an overview of SWIPT in massive distributed antenna systems. In \cite{Yongxuzhu2016}, WPT was applied to recharge UEs in massive MIMO aided $K$-tier HCNs, where UEs with large energy storage are connected to their BSs based on two typical user association schemes. Although the work in \cite{Gang_yang2015,Kai_kit_Wong_2015_Mag,Yongxuzhu2016} has studied the impact of massive MIMO antennas on WPT, they only focused on the sub-6 GHz networks.  Moreover, \cite{Gang_yang2015,Kai_kit_Wong_2015_Mag} only considered a single massive MIMO cell case, and \cite{Yongxuzhu2016} did not study the more practical case of UEs with finite battery capacity.

Existing work has studied the coverage and capacity in the mmWave cellular networks without WPT based on field measurements~\cite{TED2013IEEE_Access} or stochastic models~\cite{T_Bai_Mag_2014,Tianyang_arxiv2014,Singh_2015}, where constant transmit power was assumed. However, in the wireless powered mmWave networks, coverage and capacity need to be re-studied, since UE's transmit power becomes random and depends on the harvested energy. The use of mmWave for WPT is promising because of the fact that directed beams are used in mmWave communications and mmWave small cells will also be more densely deployed. Recent efforts on WPT have thus turned to the use of mmWave bands.  In particular, the hardware design of the mmWave rectifier circuit for WPT has been studied in, e.g., \cite{Ladan2014,Nariman2016}. The work of~\cite{Lifeng_wang2015_globcom} studied the wireless powered mmWave cellular network, in which uniform linear array with analog beamforming was implemented for WPT and uplink information transmission. Subsequently in \cite{Robert_heath_mmWave_energy}, downlink SWIPT was investigated in mmWave systems, and the average harvested energy at the  UEs and the downlink coverage probability were evaluated. However, the limitation of \cite{Lifeng_wang2015_globcom} is that it assumed that mmWave UEs have infinite battery capacity such that constant uplink transmit power can be guaranteed,  while \cite{Robert_heath_mmWave_energy} only considered WPT in the downlink and investigated the effects of WPT on downlink information transmission. To the best of our knowledge, wireless powered 5G with both sub-6 GHz and mmWave frequency bands is an open area of research.

In wireless powered cellular networks encompassing sub-6 GHz BSs and mmWave BSs equipped with their respective antenna arrays, WPT can operate at different frequencies, and a UE with finite battery capacity may harvest RF energy in the sub-6 GHz tier or the mmWave tier for accomplishing uplink information transmission. Nevertheless, such networks are not well understood. Moreover, under 5G realistic settings, how many small cells need to be deployed for supporting WPT and information transmission is still unknown.

\subsection{Contributions and Organization}
In this paper, we study wireless powered 5G dense cellular networks, in which sub-6 GHz or mmWave BSs can be selected to power UEs with finite battery capacity.  Our analysis permits to account for the key characteristics of sub-6 GHz and mmWave channels and the effects of different antenna array gains and node densities. In summary, we have made the following major contributions:
\begin{itemize}
  \item We model a wireless powered cellular network consisting of sub-6 GHz BSs and mmWave BSs equipped with antenna arrays, with the help of stochastic geometry. In the energy harvesting phase, each sub-6 GHz BS delivers energy to the nearest sub-6 GHz UE using maximum-ratio transmission (MRT) beamforming, and each mmWave BS delivers mmWave RF energy to the mmWave UE that has the minimum pathloss via narrow beam. In the uplink transmission phase, each active UE uses the harvested energy to transmit information to its associated BS.
  \item We  derive the energy coverage probability in sub-6 GHz and mmWave tiers by considering both the directed transferred power from the associated BS and the ambient RF energy from nearby BSs. We find that when the sub-6 GHz small cell density is lower than a given threshold, a UE harvests more RF energy from a mmWave BS than a sub-6 GHz BS. By considering WPT mode selection, we further derive the probability that a UE selects a sub-6 GHz BS, line-of-sight (LoS) mmWave BS or a non-LoS (NLoS) mmWave BS for WPT.

  \item Also, we derive the throughput in the uplink  sub-6 GHz and mmWave tiers with  different bandwidths. Based on the results, the number of sub-6 GHz/mmWave small cells that are required to achieve a targeted throughput is obtained. We demonstrate that the ratio $\kappa_{\mu\mathrm{UE}}$ of sub-6 GHz BS density to active sub-6 GHz UE density should be greater than a certain threshold, in order to obtain the desired performance. The throughput grows at a higher speed when increasing $\kappa_{\mu\mathrm{UE}}$, compared to increasing the number of BS antennas. Moreover, we calculate how many mmWave small cells are needed such that the achievable throughput in the mmWave tier is larger than that in the sub-6 GHz tier.

  \item Simulation results have confirmed our analysis, and illustrated that the amount of harvested energy is dominated by directed power transfer, compared to the ambient RF energy harvesting. The amount of harvested energy from ambient mmWave RF can still be larger than the sub-6 GHz counterpart in the ultra-dense mmWave tier. When power transfer mode selection is supported, the probability that a UE selects NLoS BS for WPT is negligible, and LoS mmWave WPT is also comparable to the sub-6 GHz counterpart in terms of energy coverage. It is revealed that in the dense scenario where each tier has the same number of BSs and active UEs, a mmWave UE can achieve a higher throughput than the sub-6 GHz counterpart. Furthermore, the performance gap between different mmWave frequencies increases with mmWave BS density due to the fact that more densification gains can be obtained at lower mmWave frequencies.
\end{itemize}

The remainder of the paper is organized as follows. Section II describes the system model including the energy harvesting and information transmission. Section III and Section IV analyze the energy harvesting and throughput in the considered networks, respectively.  After that, we present our simulation results in Section V. Finally, Section VI gives conclusions.

{\em Notations}---$\left(\cdot\right)^H$ denotes the conjugate transpose operator; $\mathcal{CN}\left({\bf{0},\bf{\Lambda}} \right)$ represents the complex Gaussian distribution with zero mean and covariance matrix $\bf{\Lambda}$; $\left\|\cdot\right\|$ is the Euclidean norm; $\mathbb{E}\left[ \cdot \right]$ denotes the expectation operator; ${\bf{0}}_{M \times N}$ is the $M \times N$ zero matrix, and ${\bf{I}}_M$ is the $M \times M$ identity matrix.

\section{System Descriptions}
\subsection{Network Model}
We consider a wireless powered cellular network consisting of the sub-6 GHz and mmWave small cells \footnote{In the future wireless networks such as 5G, both sub-6 GHz and mmWave frequency bands will be applied~\cite{F_Boccardi_5G}. In such networks, sub-6 GHz BSs and mmWave BSs equipped with different antenna arrays coexist, which serve UEs that operate on the sub-6 GHz or mmWave frequency bands, respectively.}, where UEs are powered by the RF energy from the BSs before uplink communication. Each sub-6 GHz BS has an array of $N$ sub-6 GHz antennas, and each mmWave BS is equipped with a large mmWave antenna array. Each sub-6 GHz UE ($\mu$UE) is equipped with a single sub-6 GHz antenna, while each mmWave UE ($\mathrm{mm}$UE) is equipped with a small mmWave antenna array, since it is expected that the shorter mmWave wavelengths would enable UEs to fit more antennas for a fixed antenna aperture. The sub-6 GHz BSs are randomly located following a homogeneous Poisson point process (HPPP) $\Phi_\mu\left(\lambda_\mu\right)$ with the density $\lambda_\mu$, and the mmWave BSs are randomly located following an independent HPPP $\Phi_\mathrm{mm}\left(\lambda_\mathrm{mm}\right)$ with the density $\lambda_\mathrm{mm}$.

When a UE requires the directed power transfer from a dedicated BS, a $\mu$UE will be connected to the sub-6 GHz BS that provides the largest received sub-6 GHz signal power, and accordingly, a $\mathrm{mm}$UE will be connected to the $\mathrm{mm}$Wave BS that provides the largest received $\mathrm{mm}$Wave signal power.

We assume that all sub-6 GHz channels are subject to independent identically distributed (IID) quasi-static Rayleigh block fading, which matches well with practical NLoS measurements~\cite{Emil_2016_small_cell_massive_MIMO,J_Park_2016}. In addition, for large number of antennas, the effect of small-scale fading is considered averaged out~\cite{Emil_2016_small_cell_massive_MIMO,ngo2013energy} and the sub-6 GHz channel power gain is dominated by random pathloss. As a consequence, Rayleigh channel distribution is suitable for modeling sub-6 GHz links when the number of antennas grows large, and has also been used in the literature such as~\cite{Emil_2016_small_cell_massive_MIMO,J_Park_2016} for studying 5G sub-6 GHz scenarios. In the mmWave systems, the high free-space mmWave pathloss leads to very limited spatial selectivity or scattering, and thus the traditional small-scale fading distributions are invalid for modeling the sparse scatting mmWave environments~\cite{Ayach2014}. As suggested in the channel measurement work~\cite{TED2013IEEE_Access},  the effect of small-scale fading in mmWave channels is omitted in this paper by considering highly directional transmissions\footnote{Note that in some existing work such as \cite{Singh_2015,Tianyang_arxiv2014}, it has been mentioned that when assuming that mmWave channel undergoes Rayleigh fading~\cite{Singh_2015} or Nakagami fading~\cite{Tianyang_arxiv2014}, the tractability of analysis can be  improved.}.

\subsection{Energy Harvesting}
\subsubsection{Sub-6 GHz Tier}  In the sub-6 GHz tier, each sub-6 GHz BS adopts MRT beamforming to transfer the energy for recharging its $\mu$UE, to maximize the transferred power. Thanks to the high diffraction and penetration characteristics of sub-6 GHz signals, the blockage effect in the sub-6 GHz channel is less significant than the mmWave counterpart~\cite{J_Park_2016,S_Rangan_2014}. To simplify our analysis, shadow fading~\cite{WeiFeng_2013_JSAC} is omitted in the sub-6 GHz tier of this paper, which is commonly-seen in the literature such as \cite{Singh_2015,J_Park_2016,Robert_heath_mmWave_energy} for tractability.   Hence, for a typical $\mu$UE, say $o$, connected with its serving sub-6 GHz BS, its instantaneous harvested power is written as
\begin{align}\label{mu_MD_receive_power}
P_r^\mu=&\underbrace{\eta_\mu{{{P_\mu}}} {\left\|\mathbf{h}_o\right\|^2}  L\left( {\left| {{X_{{o}}}} \right|}
 \right)}_{\Lambda_{\mu_1}} \nonumber\\
 &+ \underbrace{\eta_\mu \sum\limits_{k  \in {\Phi_\mu\left(\lambda_\mu\right)}\setminus\left\{ o \right\}} {{P_\mu} \left|{\mathbf{h}_{k,o} \frac{\mathbf{h}_k^H}{\left\|\mathbf{h}_k\right\|} }\right|^2 L\left( {\left| {{X_{k ,\mu}}} \right|}\right)}}_{\Lambda_{\mu_2}},
\end{align}
where $\Lambda_{\mu_1}$ is the directed transferred power, and $\Lambda_{\mu_2}$ is the total power from the ambient sub-6 GHz RF, $\eta_\mu$ is the sub-6 GHz RF-to-DC conversion efficiency, $P_\mu$ is  the transmit power of sub-6 GHz BS, $\mathbf{h}_o \sim \mathcal{CN}\left(\mathbf{0},\mathbf{I}_N\right)$ and $\left| {{X_{{o}}}} \right|$ are the small-scale fading channel vector  and distance between  the typical $\mu$UE and its serving BS, respectively,  $L\left(\left|X_{{o}}\right|\right)=\beta_\mu {\left( {{\left|X_{{o}}\right|}}\right)^{ - {\alpha_\mu}}}$ is the pathloss function with the exponent $\alpha_\mu$, where $\beta_\mu$  is a frequency dependent constant value, which is commonly set as ${(\frac{{\text{c}}}{{4\pi {f_c}}})^2}$ with $c=3 \times 10^8 \rm m/s$ and the carrier frequency $f_c$, $\frac{\mathbf{h}_k^H}{\left\|\mathbf{h}_k\right\|}$ is the MRT beamforming vector of the sub-6 GHz BS $k$ (${k  \in {\Phi_\mu}\setminus\left\{ o \right\}}$) with $\mathbf{h}_k \sim \mathcal{CN}\left(\mathbf{0},\mathbf{I}_N\right)$, $\mathbf{h}_{k,o} \sim \mathcal{CN}\left(\mathbf{0},\mathbf{I}_N\right)$ and $\left| {{X_{k,\mu}}} \right|$ are the small-scale fading channel vector and the distance between the typical $\mu$UE and the sub-6 GHz BS $k$ (except the serving sub-6 GHz BS), respectively.

\subsubsection{MmWave Tier} In the mmWave tier,  a sectored model is applied to analyze the beam pattern~\cite{A_M_Hunter_2008,Tianyang_arxiv2014,Singh_2015}, i.e.,  the effective antenna gain for a mmWave BS $\ell$ ($\ell \in \Phi_\mathrm{mm}\setminus\left\{ o \right\}$) seen by the typical $\mathrm{mm}$UE $o$  is expressed as
\begin{align}\label{array_gain_pattern}
{G_\ell} = \left\{ \begin{array}{ll}
{\mathrm{M_B}\mathrm{M_D}}{\rm{,}} &{\Pr _{\mathrm{M_\mathrm{B} M_\mathrm{D}}}}{\rm{ = }}{ {\frac{\theta_\mathrm{B}\theta_\mathrm{D}}{4\pi^2}}},\\
{\mathrm{M_B}}{\mathrm{m_D}}{\rm{,}} &{\Pr _{\mathrm{M_B m_D}}}{\rm{ = }}\frac{{\theta_\mathrm{B} \left( {2\pi  - \theta_\mathrm{D}} \right)}}{{{{4\pi^2}}}},\\
{\mathrm{m_B}}{\mathrm{M_D}}{\rm{,}} &{\Pr _{\mathrm{m_B M_D}}}{\rm{ = }}\frac{{ \left( {2\pi  - \theta_\mathrm{B}} \right)\theta_\mathrm{D}}}{{{{4\pi^2}}}},\\
{\mathrm{m_B} \mathrm{m_D}}, &{\Pr _{\mathrm{m_B m_D}}} = \frac{\left( {{2\pi- \theta_\mathrm{B}}}\right)\left({2\pi-\theta_\mathrm{D}} \right)}{{4\pi^2}},
\end{array} \right.
\end{align}
where $\mathrm{M}_\mathrm{B}$, $\mathrm{m}_\mathrm{B}$, and $\theta_\mathrm{B}$ are the main lobe gain, side lobe gain, and half power beamwidth of the mmWave BS antenna, respectively, and  $\mathrm{M}_\mathrm{D}$, $\mathrm{m}_\mathrm{D}$, and $\theta_\mathrm{D}$ are the main lobe gain, side lobe gain, and half power beamwidth of the $\mathrm{mm}$UE antenna, respectively. We assume that the maximum array gain ${\mathrm{M_B}\mathrm{M_D}}$ can be obtained for the typical BS and its $\mathrm{mm}$UE.

Recognizing that mmWave communication is sensitive to the blockage in the outdoor scenario, a $\mathrm{mm}$UE is associated with either a LoS mmWave BS or a NLoS mmWave BS.  We denote $f_\mathrm{Pr}\left(R\right)$ as the probability that a link at a distance $R$ is LoS, so that the NLoS probability of a link is $1-f_\mathrm{Pr}\left(R\right)$. We consider two different pathloss laws: $L\left(R\right)=\beta^{\rm{mm}}_\mathrm{LoS} R^{-\alpha^{\rm{mm}}_\mathrm{LoS}}$ is the pathloss function for LoS channel and $L\left(R\right)=\beta^{\rm{mm}}_\mathrm{NLoS} R^{-\alpha^{\rm{mm}}_\mathrm{NLoS}}$ is the pathloss function for NLoS channel, where $\beta^{\rm{mm}}_\mathrm{LoS}$, $\beta^{\rm{mm}}_\mathrm{NLoS}$ are the frequency dependent constant values and $\alpha^{\rm{mm}}_\mathrm{LoS}$, $\alpha^{\rm{mm}}_\mathrm{NLoS}$ are the pathloss exponents.

 For a typical $\mathrm{mm}$UE $o$ connected with its serving mmWave BS, its instantaneous harvested power is written as
\begin{align}\label{mm_MD_receive_power}
P_r^\mathrm{mm}=&\underbrace{\eta_\mathrm{mm}{{{P_\mathrm{mm}}}} \mathrm{M_\mathrm{B}}\mathrm{M_\mathrm{D}} L\left( {\left| {{Y_{{o}}}} \right|}
 \right)}_{\Lambda_{\mathrm{mm}_1}} \nonumber\\
 &\quad+\underbrace{\eta_\mathrm{mm}\sum\limits_{\ell  \in {\Phi_\mathrm{mm}\left(\lambda_\mathrm{mm}\right)}\setminus\left\{ o \right\}} {{P_\mathrm{mm}} {G_\ell} L\left( {\left| {{Y_{\ell,\mathrm{mm}}}} \right|}\right)}}_{\Lambda_{\mathrm{mm}_2}},
\end{align}
where $\Lambda_{\mathrm{mm}_1}$ is the directed transferred power, and  $\Lambda_{\mathrm{mm}_2}$ is the total power from the ambient mmWave RF, $\eta_\mathrm{mm}$ is the mmWave RF-to-DC conversion efficiency, $P_\mathrm{mm}$ is the mmWave BS transmit power, $\left| {{Y_{{o}}}} \right|$ is the distance between the typical $\mathrm{mm}$UE and its serving mmWave BS, and $\left| {{Y_{\ell,\mathrm{mm}}}} \right|$ is the distance between the typical $\mathrm{mm}$UE and the mmWave BS $\ell \in {\Phi_\mathrm{mm}}\setminus\left\{ o \right\}$ (except the serving mmWave BS).

\subsection{Uplink Transmission}
\subsubsection{Sub-6 GHz Tier} In the sub-6 GHz tier, maximum-ratio combining (MRC) is utilized for maximizing the received signal power at the sub-6 GHz BS. The receive signal-to-interference-plus-noise ratio (SINR) of a typical sub-6 GHz BS from its intended $\mu$UE is therefore given by
\begin{align}\label{muWave_SINR_original}
\mathrm{SINR}_{\mu}&=\frac{S_{\left(\lambda_\mu\right)}^\mu}{I_{\left(\widetilde{\lambda}_{\mu\mathrm{UE}}\right)}^\mu+{\sigma^2}},
\end{align}
where
 \begin{equation}\label{SINR_part1}
 \left\{\begin{aligned}
S_{\left(\lambda_\mu\right)}^\mu &=P_{\mathrm{UE}_o}^{\mu} \left\|\mathbf{g}_o\right\|^2 L\left({\left| {{X_{{o}}}} \right|}\right),\\
 I_{\left(\widetilde{\lambda}_{\mu\mathrm{UE}}\right)}^\mu &=\sum\limits_{i  \in {\widetilde{\Phi}_{\mu\mathrm{UE}}\left(\widetilde{\lambda}_{\mu\mathrm{UE}}\right)}\setminus\left\{ o \right\}} {{P_{\mathrm{UE}_i}^{\mu}} \left|{\frac{\mathbf{g}_o^H}{\left\|\mathbf{g}_o\right\|} \mathbf{g}_i }\right|^2 L\left( {\left| {{X_{i ,\mu}}} \right|}\right)}.
\end{aligned}\right.
\end{equation}
In \eqref{SINR_part1}, $P_{\mathrm{UE}_i}^\mu$ denotes the $i$-th $\mu$UE's transmit power, $\mathbf{g}_o \sim \mathcal{CN}\left(\mathbf{0},\mathbf{I}_N\right)$ is the small-scale fading channel vector between the typical sub-6 GHz BS  and its intended $\mu$UE, $\widetilde{\Phi}_{\mu\mathrm{UE}}(\widetilde{\lambda}_{\mu\mathrm{UE}})$ is the point process for the active $\mu$UEs with density $\widetilde{\lambda}_{\mu\mathrm{UE}}$, $\mathbf{g}_i \sim \mathcal{CN}\left(\mathbf{0},\mathbf{I}_N\right)$ and $\left| {{X_{i ,\mu}}} \right|$ are the small-scale fading channel vector and the distance between the typical sub-6 GHz BS and interfering $\mu$UE $i$, respectively, and $\sigma^2$ is the noise power.

\subsubsection{MmWave Tier} In the mmWave tier, we only consider the LoS uplink transmissions, since each mmUE uses lower transmit power from  limited harvested energy and the harvested energy from NLoS link is much lower, which means that NLoS uplink will be blocked.  According to the LoS mmWave model in~\cite{Tianyang_arxiv2014,J_Park_2016},  the receive SINR of a typical mmWave BS from its intended mmUE is given by
\begin{align}\label{mmWave_SINR_original}
\mathrm{SINR}_{\mathrm{mm}}=\frac{S_{\left(\lambda_\mathrm{mm}\right)}^{\mathrm{mm}}}{
I_{\left(\widetilde{\lambda}_{\mathrm{mmUE}}\right)}^{\mathrm{mm}}+{\sigma^2}},
\end{align}
where
 \begin{equation}\label{SINR_part2_mmwave}
 \left\{\begin{aligned}
S_{\left(\lambda_\mathrm{mm}\right)}^{\mathrm{mm}} &=P_{\mathrm{UE}_o}^{\mathrm{mm}} \mathrm{M}_D \mathrm{M}_B L\left({\left| {{Y_{{o}}}} \right|}\right)\mathbf{1}\left(\left| {{Y_{{o}}}} \right|< R_\mathrm{LoS}\right),\\
 I_{\left(\widetilde{\lambda}_{\mathrm{mmUE}}\right)}^{\mathrm{mm}}&=\sum\limits_{j  \in {\widetilde{\Phi}_{\mathrm{mmUE}}}\left( \widetilde{\lambda}_{\mathrm{mmUE}}\right)\setminus\left\{ o \right\}} P_{\mathrm{UE}_j}^{\mathrm{mm}} \widetilde{G}_j L_\mathbf{1}\left( {\left| {{Y_{j ,\mu}}} \right|}\right)
\end{aligned}\right.
\end{equation}
with $L_\mathbf{1}\left( {\left| {{Y_{j ,\mu}}} \right|}\right)=L\left( {\left| {{Y_{j ,\mu}}} \right|}\right)\mathbf{1}\left(\left| {{Y_{j ,\mu}}} \right|< R_\mathrm{LoS}\right)$. Here, $P_{\mathrm{UE}_o}^{\mathrm{mm}}$ is the typical mmUE's transmit power and $P_{\mathrm{UE}_j}^{\mathrm{mm}}$ is the $j$-th interfering mmUE's transmit power, $\mathbf{1}\left(\cdot\right)$ is the indicator function that returns one if the condition is satisfied and zero otherwise, $R_\mathrm{LoS}$ denotes the maximum distance that LoS can be guaranteed~\cite{Tianyang_arxiv2014,J_Park_2016} (i.e., $f_\mathrm{Pr}\left(R\right)=1$ as $R \leq R_\mathrm{LoS}$ and otherwise $f_\mathrm{Pr}\left(R\right)=0$.), $\widetilde{\Phi}_{\mathrm{mmUE}}(\widetilde{\lambda}_{\mathrm{mmUE}})$ is the point process corresponding to the active mmUEs with the density $\widetilde{\lambda}_{\mathrm{mmUE}}$, $\widetilde{G}_j$ is the effective antenna gain for an interfering mmUE $j$ seen by the typical mmWave BS, which follows the distribution given in \eqref{array_gain_pattern}, and $\left| {{Y_{j ,\mu}}} \right|$ is the distance between the interfering mmUE $j$ and the typical mmWave BS.

\section{Energy Harvesting}
We evaluate the wireless energy harvesting in the sub-6 GHz and mmWave cellular networks. To gain comprehensive understanding, we respectively examine the directed transferred power from the associated BS  and the ambient RF harvested power from nearby BSs that a UE can obtain, which allows us to quantify the harvested energy from the dedicated RF and ambient RF.  Note that the minimum amount of  energy is required to activate the harvesting circuit, which depends on the specific circuit designs based on CMOS, HSMS, SMS and etc. As surveyed in \cite{X_Lu_Survey_2015}, the minimum RF input power for CMOS technology can be as low as $-27$ dBm  based on the prior circuit work from 2006 to 2014. Therefore, the RF energy harvesting sensitivity level is much lower and can be omitted~\cite{Kaibin2014,CheDZ14}.
In fact, since we consider 5G dense cellular networks, where dense BSs  act as dedicated RF energy sources to power UEs via narrow beams, the amount of the received energy at a UE will be much larger than the RF energy harvesting sensitivity level~\cite{X_Lu_Survey_2015}.

\subsection{Directed Transferred Power}
\subsubsection{Sub-6 GHz Tier} In the sub-6 GHz tier, given a power threshold $P_\mathrm{th}$, the coverage probability that the directed transferred power is larger than $P_\mathrm{th}$  can be derived as
\begin{align}\label{DWPT_u_1}
\overline{\Psi}_\mathrm{D}^{\mu}\left(P_\mathrm{th}\right)=& \sum\limits_{n = 0}^{N - 1}  {\left({\frac{P_\mathrm{th}}{\eta_\mu P_\mu \beta_\mu}}\right)^n}\frac{2 \pi \lambda_\mu}{n!} \nonumber\\
&\qquad \times \int_0^\infty {{{{e^{ - \frac{r^{\alpha_\mu}P_\mathrm{th}}{\eta_\mu P_\mu \beta_\mu}-\pi \lambda_\mu r^2}}}}} r^{\alpha_\mu n+1}  dr.
\end{align}
\begin{proof}
Based on \eqref{mu_MD_receive_power},  $\overline{\Psi}_\mathrm{D}^{\mu}\left(P_\mathrm{th}\right)$ is calculated as
\begin{align}\label{pf1_1_1}
&\overline{\Psi}_\mathrm{D}^{\mu}\left(P_\mathrm{th}\right)=\Pr\left(\Lambda_{\mu_1} > P_\mathrm{th} \right) \nonumber\\
&=\int_0^\infty \Pr\left(\eta_\mu{{{P_\mu}}}{\left\|\mathbf{h}_o\right\|^2} \beta_\mu r^{-\alpha_\mu} > P_\mathrm{th} \right) f_{\left| {{X_{{o}}}} \right|} \left( r \right) dr,
\end{align}
where $f_{\left| {{X_{{o}}}} \right|} \left( r \right)$ is the probability density function (PDF) of {the distance between a sub-6 GHz UE and its nearest sub-6 GHz BS}, which is given by
\begin{align}\label{pf1_2}
f_{\left| {{X_{{o}}}} \right|} \left( r \right)= 2 \pi \lambda_\mu r \exp\left(-\pi \lambda_\mu r^2\right).
\end{align}
Considering that $\left\|\mathbf{h}_o\right\|^2 \sim \Gamma\left(N,1\right)$, we can further calculate \eqref{pf1_1_1} as
\begin{align}\label{pf1_3}
&\overline{\Psi}_\mathrm{D}^{\mu}\left(P_\mathrm{th}\right)= \sum\limits_{n = 0}^{N - 1}
 \int_{0}^\infty {\frac{{{e^{ - \frac{r^{\alpha_\mu}P_\mathrm{th}}{\eta_\mu P_\mu \beta_\mu} }}}}{{n!}}} {\left({\frac{r^{\alpha_\mu}P_\mathrm{th}}{\eta_\mu P_\mu \beta_\mu}}\right)^n}   f_{\left| {{X_{{o}}}} \right|} \left( r \right) dr. 
\end{align}
Substituting \eqref{pf1_2} into \eqref{pf1_3}, we obtain \eqref{DWPT_u_1}.
\end{proof}
Based on \eqref{DWPT_u_1}, the sufficient condition for $\overline{\Psi}_\mathrm{D}^{\mu}\left(P_\mathrm{th}\right)>\varepsilon$ ($0<\varepsilon < 1$) for a given $P_\mathrm{th}$, is given by the following corollary.

\begin{corollary}
The probability of the achievable directed transferred power $P_\mathrm{th}$ is larger than $\varepsilon$, if the sub-6 GHz BS density  satisfies
\begin{align}\label{col_1_11}
\lambda_\mu > \omega_o \left(\frac{\eta_\mu P_\mu \beta_\mu N}{P_\mathrm{th}}\right)^{-2/\alpha_\mu}
\end{align}
with large $N$, where $\omega_o=\ln\left(\frac{1}{1-\varepsilon}\right)/\pi$.
\end{corollary}
\begin{proof}
With large $N$, $\left\|\mathbf{h}_o\right\|^2 \approx N$~\cite{Lifeng_Massive_MIMO}, and \eqref{pf1_1_1} can be approximated as
\begin{align}\label{mu_BS_approx}
\widetilde{\Psi}_\mathrm{D}^{\mu}\left(P_\mathrm{th}\right)&= \int_0^{\left(\frac{\eta_\mu P_\mu \beta_\mu N}{P_\mathrm{th}}\right)^{1/\alpha_\mu}} f_{\left| {{X_{{o}}}} \right|} \left( r \right) dr \nonumber\\
&=1-\exp\left(-\pi \lambda_\mu \left(\frac{\eta_\mu P_\mu \beta_\mu N}{P_\mathrm{th}}\right)^{2/\alpha_\mu}\right).
\end{align}
Letting $\widetilde{\Psi}_\mathrm{D}^{\mu}\left(P_\mathrm{th}\right) > \varepsilon$, after manipulations, gives \eqref{col_1_11}.
\end{proof}

It is indicated from \textbf{Corollary 1} that the number of sub-6 GHz BSs needs to be large enough for WPT. Moreover, the required BS density decreases when adding more BS antennas, due to the fact that the decreased densification gains can be redeemed by obtaining more antenna gains.


\subsubsection{MmWave Tier} In the mmWave tier, the coverage probability that the directed transferred power is larger than a threshold $P_\mathrm{th}$ can be derived as
 \begin{align}\label{mmWave_pf1}
&\overline{\Psi}_\mathrm{D}^{\mathrm{mm}}\left(P_\mathrm{th}\right)=2\pi \lambda_\mathrm{mm}\bigg[ \nonumber\\
&\quad \int_0^{\varsigma\left(\beta^{\rm{mm}}_\mathrm{LoS},\alpha^{\rm{mm}}_\mathrm{LoS}\right)} y {f_{\Pr }}
\left( y\right){e^{ - 2\pi \lambda_\mathrm{mm} \left[ {\Theta \left( y \right) + \Xi \left( {{\Im_\mathrm{LoS}}
} \right)} \right]}} dy + \nonumber\\
&~~\int_0^{\varsigma\left(\beta^{\rm{mm}}_\mathrm{NLoS},\alpha^{\rm{mm}}_\mathrm{NLoS}\right)} y{\left(1-f_{\Pr }\left( y \right)\right)}
{e^{ - 2\pi \lambda_\mathrm{mm} \left[ \Theta \left(\Im_\mathrm{NLoS}\right)+ {\Xi \left( y\right)} \right]}} dy\bigg],
\end{align}
where $\varsigma\left(y_1,y_2\right)=\left(\frac{\eta_\mathrm{mm}{{{P_\mathrm{mm}}}} \mathrm{M_\mathrm{B}}\mathrm{M_D}}{P_\mathrm{th}}y_1\right)^{1/y_2}$, $\Theta \left( y \right) = \int_0^y {t f_{\Pr}\left( t \right)dt}$,
$\Xi \left( y\right) = \int_0^y {\left( {1 - {f_{\Pr }}\left( t \right)} \right)tdt}$, $\Im_\mathrm{LoS}=
\left( {\frac{{{\beta^{\rm{mm}}_\mathrm{NLoS}}}}{{{\beta^{\rm{mm}}_\mathrm{LoS}}}}} \right)^{1/\alpha^{\rm{mm}}_\mathrm{NLoS}}
y^{{\alpha^{\rm{mm}}_\mathrm{LoS}}/{\alpha^{\rm{mm}}_\mathrm{NLoS}}}$, and $\Im_\mathrm{NLoS}=
\left( {\frac{{{\beta^{\rm{mm}}_\mathrm{LoS}}}}{{{\beta^{\rm{mm}}_\mathrm{NLoS}}}}} \right)^{1/\alpha^{\rm{mm}}_\mathrm{LoS}}
y^{{\alpha^{\rm{mm}}_\mathrm{NLoS}}/{\alpha^{\rm{mm}}_\mathrm{LoS}}}$.

\begin{proof}
Based on \eqref{mm_MD_receive_power}, $\overline{\Psi}_\mathrm{D}^{\mathrm{mm}}\left(P_\mathrm{th}\right)$ is calculated as
\begin{align}\label{pf1_1}
\overline{\Psi}_\mathrm{D}^{\mathrm{mm}}\left(P_\mathrm{th}\right)&=\Pr\left(\Lambda_{\mathrm{mm}_1}> P_\mathrm{th}  \right) \nonumber\\
&=\Upsilon_\mathrm{LoS} \int_0^{\varsigma\left(\beta^{\rm{mm}}_\mathrm{LoS},\alpha^{\rm{mm}}_\mathrm{LoS}\right)} f_R^{\mathrm{LoS}}\left( y \right) dy  \nonumber\\
&+   \Upsilon_\mathrm{NLoS} \int_0^{\varsigma\left(\beta^{\rm{mm}}_\mathrm{NLoS},\alpha^{\rm{mm}}_\mathrm{NLoS}\right)}  f_R^{\mathrm{NLoS}}\left( y \right) dy,
\end{align}
where  $\Upsilon_\mathrm{LoS}$ represents the probability that the typical user is connected to a LoS BS and $\Upsilon_\mathrm{NLoS}=1-\Upsilon_\mathrm{LoS}$ represents the probability that the typical user is connected to a NLoS BS,  $f_R^{\mathrm{LoS}}\left( x \right)$ is the conditioned PDF of the distance  between the typical $\mathrm{mm}$UE and its serving LoS BS, and $f_R^{\mathrm{NLoS}}\left( x \right)$ is the conditioned PDF of the distance  between the typical $\mathrm{mm}$UE and its serving NLoS BS, which are given by~\cite{Tianyang_arxiv2014}
\begin{align}\label{LoS_Distance_conditionPDF}
f_R^{\mathrm{LoS}}\left( y \right) = \frac{{2\pi \lambda_\mathrm{mm} }}{{{\Upsilon_\mathrm{LoS}}}}y {f_{\Pr }}
\left( y\right){e^{ - 2\pi \lambda_\mathrm{mm} \left[ {\Theta \left( y \right) + \Xi \left( {{\Im_\mathrm{LoS}}
} \right)} \right]}},
\end{align}
and
\begin{align}\label{NLoS_Distance_conditionPDF}
 f_R^{\mathrm{NLoS}}\left( y\right) = \frac{{2\pi \lambda_\mathrm{mm} }}{{{\Upsilon_\mathrm{NLoS}}}}y{\left(1-f_{\Pr }\left( y \right)\right)}
{e^{ - 2\pi \lambda_\mathrm{mm} \left[ \Theta \left(\Im_\mathrm{NLoS}\right)+ {\Xi \left( y\right)} \right]}},
\end{align}
respectively. Substituting \eqref{LoS_Distance_conditionPDF} and \eqref{NLoS_Distance_conditionPDF} into \eqref{pf1_1} yields \eqref{mmWave_pf1}.
\end{proof}

\begin{corollary}
Given a threshold $P_\mathrm{th}$, the coverage probability of the directed power transfer in the mmWave tier is larger than that in the sub-6 GHz tier under the following condition:
 \begin{align}\label{corollary2}
 \lambda_\mu < \frac{\ln\left(1-\overline{\Psi}_\mathrm{D}^{\mathrm{mm}}\left(P_\mathrm{th}\right)\right)^{-1}}{\pi  \left(\frac{\eta_\mu P_\mu \beta N}{P_\mathrm{th}}\right)^{2/\alpha_\mu}}.
 \end{align}
 In particular, when only LoS mmWave links with $f_\mathrm{Pr}\left(R\right)=\mathbf{1}\left(R \leq R_\mathrm{LoS}\right)$ are used for power transfer, the above condition is transformed as
 \begin{align}\label{col_2_2}
\lambda_\mu<\lambda_\mathrm{mm} {\xi^{2}}{\left(\frac{\eta_\mu P_\mu \beta N}{P_\mathrm{th}}\right)^{-2/\alpha_\mu}},
 \end{align}
where $\xi=\min\left\{R_\mathrm{LoS},\varsigma\left(\beta^{\rm{mm}}_\mathrm{LoS},\alpha^{\rm{mm}}_\mathrm{LoS}\right)\right\}$.
\end{corollary}

\begin{proof}
Based on \textbf{Corollary 1} and \eqref{mmWave_pf1}, we can directly obtain \eqref{corollary2}. By considering the LoS mmWave model with $f_\mathrm{Pr}\left(R\right)=\mathbf{1}\left(R \leq R_\mathrm{LoS}\right)$~\cite{Tianyang_arxiv2014,J_Park_2016}, the probability that the directed transferred power is larger than a threshold $P_\mathrm{th}$ under LoS mmWave power transfer is found as
\begin{align}\label{mmWave_th_P}
&\overline{\Psi}_\mathrm{D}^{\mathrm{mm}}\left(P_\mathrm{th}\right) \nonumber\\
&=\Pr\left(\eta_\mathrm{mm}{{{P_\mathrm{mm}}}} \mathrm{M_B}\mathrm{M_D} L\left( {\left| Y_o \right|}
 \right)\mathbf{1}\left(\left| Y_o \right|<R_\mathrm{LoS}\right)> P_\mathrm{th}  \right) \nonumber\\
 &\mathop = \limits^{(\rm{a})}\int_0^{\xi} 2 \pi \lambda_\mathrm{mm} r \exp\left(-\pi \lambda_\mathrm{mm} r^2\right) dr\nonumber\\
 &=1-\exp\left(-\pi \lambda_\mathrm{mm} \xi^2 \right),
\end{align}
where $\xi=\min\left\{R_\mathrm{LoS},\varsigma\left(\beta^{\rm{mm}}_\mathrm{LoS},\alpha^{\rm{mm}}_\mathrm{LoS}\right)\right\}$, step (a) is obtained by considering the fact that UEs try to be connected to the nearest BSs such that there exist LoS links. Substituting \eqref{mmWave_th_P} into \eqref{corollary2}, we obtain  \eqref{col_2_2} and complete the proof.
\end{proof}

Based on \textbf{Corollary 2}, we find that the mmWave tier can achieve better energy coverage than the sub-6 GHz tier, when the scale of sub-6 GHz tier is lower than the right-hand-side (RHS) of \eqref{corollary2}.

\subsection{Ambient RF Harvested Power}
In order to avoid singularity at zero distance and ensure the finite moments of the sum of the ambient RF signals, we incorporate the distance constraint into the path loss function in this subsection, which is $\tilde{L}\left(\left|X\right|\right)=\beta { {{\left( \max\left(r_o,\left|X\right|\right) \right) }}^{ - {\alpha}}}$ with a reference distance $r_o $~\cite{Baccelli2006_gen,Kaibin2014}. It should be noted that the reference distance has negligible effect on the evaluation of the energy coverage probability, since the harvested energy is usually larger than the predefined threshold when $\left|X\right|\leq r_o$.

\subsubsection{Sub-6 GHz Tier}
In the sub-6 GHz tier, let $\overline{\Psi}_\mathrm{A}^{mm}\left(P_\mathrm{th}\right)$ denote the probability that the ambient RF harvested power is larger than a threshold $P_\mathrm{th}$, $\overline{\Psi}_\mathrm{A}^{\mu}\left(P_\mathrm{th}\right)$ is upper bounded as
\begin{align}\label{mmWave_density_value1223}
&\overline{\Psi}_\mathrm{A}^{\mu}\left(P_\mathrm{th}\right) \leq \min\left\{{\mathrm{var}\left[\Xi\right]}/{\left(\frac{P_\mathrm{th}}{\eta_\mu {P_\mu}}-\mathbb{E}\left[\Xi\right]\right)^2},1\right\},
\end{align}
where $\mathbb{E}\left[\Xi\right]$ and $\mathrm{var}\left[\Xi\right]$ are respectively given by \eqref{Exp_Xi_1999} and \eqref{var_Xi_2000} at the top of this page. In \eqref{Exp_Xi_1999} and \eqref{var_Xi_2000}, $E_{(n)}(z)=\int_1^\infty e^{-zt} t^{-n} dt$ is the exponential integral~\cite{gradshteyn}.
\begin{figure*}[htbp]
\normalsize
  \begin{align}\label{Exp_Xi_1999}
& \mathbb{E}\left[\Xi\right]=\beta_\mu 2\pi \lambda_\mu \left(\frac{{r_o^{ - {\alpha _\mu }}}}{{\rm{2}}}\left( {r_{\rm{o}}^2 - {{\left( {\pi {\lambda _\mu }} \right)}^{ - 1}}\left( {1 - {e^{ - \pi {\lambda _\mu }r_{\rm{o}}^2}}} \right)} \right)+\frac{{r_o^{2 - {\alpha _\mu }}}}{{{\alpha _\mu } - 2}} - \frac{{r_o^{2 - {\alpha _\mu }}}}{2}{E_{\left( {\frac{{{\alpha _\mu }}}{2}} \right)}}\left( {\pi {\lambda _\mu }r_o^2} \right)\right)
 \end{align}
\hrulefill \vspace*{0pt}
 \begin{align}\label{var_Xi_2000}
&\mathrm{var}\left[\Xi\right]=\beta_\mu^2 4 \pi \lambda_\mu \left(\frac{{r_o^{ - 2{\alpha _\mu }}}}{{\rm{2}}}\left( {r_{\rm{o}}^2 - {{\left( {\pi {\lambda _\mu }} \right)}^{ - 1}}\left( {1 - {e^{ - \pi {\lambda _\mu }r_{\rm{o}}^2}}} \right)} \right)+\frac{{r_o^{2 - 2{\alpha _\mu }}}}{{2{\alpha _\mu } - 2}} - \frac{{r_o^{2 - 2{\alpha _\mu }}}}{2}{E_{\left( {{\alpha _\mu }} \right)}}\left( {\pi {\lambda _\mu }r_o^2} \right) \right)
 \end{align}
\hrulefill \vspace*{0pt}
\end{figure*}
\begin{proof}
Based on \eqref{mu_MD_receive_power}, $\overline{\Psi}_\mathrm{A}^\mu\left(P_\mathrm{th}\right)$ is calculated as
\begin{align}\label{AmWPT_mmWave_1}
&\overline{\Psi}_\mathrm{A}^\mu\left(P_\mathrm{th}\right)=\Pr\left(\Lambda_{\mu_2} > P_\mathrm{th} \right) \nonumber\\
&=\Pr\Big( \underbrace{ \sum\limits_{k  \in {\Phi_\mu\left(\lambda_\mu\right)}\setminus\left\{ o \right\}} { \left|{\mathbf{h}_{k,o} \frac{\mathbf{h}_k^H}{\left\|\mathbf{h}_k\right\|} }\right|^2 \tilde{L}\left( {\left| {{X_{k ,\mu}}} \right|}\right)}}_{\Xi} > \frac{P_\mathrm{th}}{\eta_\mu {P_\mu}} \Big) \nonumber\\
&\mathop  \le \limits^{(b)} \min\left\{\frac{\mathrm{var}\left[\Xi\right]}{\left(\frac{P_\mathrm{th}}{\eta_\mu {P_\mu}}-\mathbb{E}\left[\Xi\right]\right)^2},1\right\},
\end{align}
where step (b) is from the Chebyshev's inequality. Here, $\mathbb{E}\left[\Xi\right]$ and $\mathrm{var}\left[\Xi\right]$ denote  the expectation and variance of $\Xi$, respectively.

We first derive $\mathbb{E}\left[\Xi\right]$. By using the Campbell's theorem~\cite{Baccelli2009}, $\mathbb{E}\left[\Xi\right]$ is given by
\begin{align}\label{Xi_exp_cam}
\mathbb{E}\left[\Xi\right]=&\mathbb{E}\left[\left|{\mathbf{h}_{k,o} \frac{\mathbf{h}_k^H}{\left\|\mathbf{h}_k\right\|} }\right|^2\right]
\beta_\mu   2\pi \lambda_\mu \nonumber\\
&\times \int_0^\infty \int_r^\infty \left(\max\left(r_o,t\right)\right)^{-\alpha_\mu} t f_{\left| {{X_{{o}}}} \right|} \left( r \right) dt dr,
\end{align}
where $f_{\left| {{X_{{o}}}} \right|} \left( r \right)$ is given by \eqref{pf1_2}. Since ${\mathbf{h}_{k,o} \frac{\mathbf{h}_k^H}{\left\|\mathbf{h}_k\right\|} }$ follows complex Gaussian distribution and is independent of $\mathbf{h}_k$, we have $\left|{\mathbf{h}_{k,o} \frac{\mathbf{h}_k^H}{\left\|\mathbf{h}_k\right\|} }\right|^2 \sim \exp\left(1\right)$, thus $\mathbb{E}\left[\left|{\mathbf{h}_{k,o} \frac{\mathbf{h}_k^H}{\left\|\mathbf{h}_k\right\|} }\right|^2\right]=1$. By changing the order of integration, we further derive \eqref{Xi_exp_cam} as
\begin{align}\label{Xi_exp_cam_111}
&\mathbb{E}\left[\Xi\right]=\beta_\mu  2\pi \lambda_\mu  \int_0^\infty  \left(\max\left(r_o,t\right)\right)^{-\alpha_\mu} t \left( {1 - \exp \left( { - \pi {\lambda _\mu }{t^2}} \right)} \right)dt.
\end{align}
After calculating the integral in \eqref{Xi_exp_cam_111}, we obtain \eqref{Exp_Xi_1999}.

The variance of $\Xi$ is derived as
\begin{align}\label{Xi_calculation_2016}
&\mathrm{var}\left[\Xi\right]={\left. {\frac{{{\partial ^2}}}{{\partial {s^2}}}\mathbb{E}\left[ {\exp \left( {s\Xi } \right)} \right]} \right|_{s{\rm{ = }}0}} - {\left( {\mathbb{E}\left[ \Xi  \right]} \right)^2}\nonumber\\
& \mathop = \limits^{(c)} \mathbb{E}\left[\left|{\mathbf{h}_{k,o} \frac{\mathbf{h}_k^H}{\left\|\mathbf{h}_k\right\|} }\right|^4\right]
\beta_\mu^2  2\pi \lambda_\mu \nonumber\\
&~~\times\int_0^\infty \int_r^\infty \left(\max\left(r_o,t\right)\right)^{-2\alpha_\mu} t f_{\left| {{X_{{o}}}} \right|} \left( r \right) dt dr,
\end{align}
where step (c) is obtained by using the modified Campbell's theorem. Since $\mathbb{E}\left[\left|{\mathbf{h}_{k,o} \frac{\mathbf{h}_k^H}{\left\|\mathbf{h}_k\right\|} }\right|^4\right]=2$, we can finally obtain \eqref{var_Xi_2000} and complete the proof. \end{proof}

\subsubsection{MmWave Tier} In the mmWave tier, let $\overline{\Psi}_\mathrm{A}^{\mathrm{mm}}\left(P_\mathrm{th}\right)$ denote the probability that the ambient RF harvested power is larger than a threshold $P_\mathrm{th}$, given $\varepsilon$, $\overline{\Psi}_\mathrm{A}^{\mathrm{mm}}\left(P_\mathrm{th}\right)>\varepsilon$ holds when the mmWave BS density meets
\begin{align}\label{mu_BS_density_1}
 \lambda_\mathrm{mm} > \left(\frac{P_\mathrm{th}}{ \vartheta_\mathrm{mm} \eta_\mathrm{mm} {P_\mathrm{mm}} }\right)^{\frac{2}{\alpha^{\rm{mm}}_\mathrm{LoS}}},
\end{align}
where $ \vartheta_\mathrm{mm}$ is a constant value defined by
\begin{align}\label{varpai_u}
\Pr\Big(\sum\limits_{\ell  \in {\Phi_{\rm{mm}} \left(1\right)}\setminus\left\{ o \right\}} {G_\ell}\Delta(\left|Y_{\ell ,\rm{mm}}\right|)>\vartheta_\mathrm{mm}\Big)=\varepsilon
\end{align}
with $\Delta(\left|Y_{\ell ,\rm{mm}}\right|)=\tilde{L}_\mathrm{LoS}\left(\left|Y_{\ell ,\rm{mm}}\right|\right)\amalg(f_\mathrm{Pr}\left(\left|Y_{\ell ,\rm{mm}}\right|\right))$, where $\amalg(x)$ represents the Bernoulli distribution.

\begin{proof}
Based on \eqref{mm_MD_receive_power}, the probability that the ambient RF harvested power $\Lambda_{\mathrm{mm}_2}$ is larger than a threshold $P_\mathrm{th}$  can be obtained as
\begin{align}\label{DWPT_u_1_Ambient}
&\overline{\Psi}_\mathrm{A}^{\rm{mm}}\left(P_\mathrm{th}\right)=\Pr\left(\Lambda_{\mathrm{mm}_2}> P_\mathrm{th} \right)\nonumber\\
&=\Pr\left(\sum\limits_{ \ell \in {\Phi_\mathrm{mm}\left(\lambda_\mathrm{mm}\right)}\setminus\left\{ o \right\}} {G_\ell} \tilde{L}\left( {\left| {{Y_{\ell,\mathrm{mm}}}} \right|}\right) >\frac{ P_\mathrm{th}}{ \eta_\mathrm{mm} P_\mathrm{mm}} \right).
\end{align}
Since the ambient RF energy from the NLoS BSs is negligible, $\overline{\Psi}_\mathrm{A}^{\rm{mm}}\left(P_\mathrm{th}\right)$ can be lower-bounded as
\begin{align}\label{LoS_Ambient_RF}
&\overline{\Psi}_\mathrm{A}^{\rm{mm}}\left(P_\mathrm{th}\right) \geq \Pr\Big(\sum\limits_{ \ell \in {\Phi_\mathrm{mm}\left(\lambda_\mathrm{mm}\right)}\setminus\left\{ o \right\}} {G_\ell}\Delta(\left|Y_{\ell ,\rm{mm}}\right|)>\frac{ P_\mathrm{th}}{ \eta_\mathrm{mm} P_\mathrm{mm}} \Big)\nonumber\\
&~\qquad\mathop = \limits^{(d)} \Pr\Big( \sum\limits_{ \ell \in {\Phi_\mathrm{mm}\left(1\right)}\setminus\left\{ o \right\}} {G_\ell}\Delta(\left|Y_{\ell ,\rm{mm}}\right|)>\lambda_\mathrm{mm}^{-\frac{\alpha^{\rm{mm}}_\mathrm{LoS}}{2}}\frac{ P_\mathrm{th}}{ \eta_\mathrm{mm} P_\mathrm{mm}} \Big),
\end{align}
where step (d) is obtained by using the Mapping theorem.

Given $\varepsilon$, we define the constant $\vartheta_\mu$ as \eqref{varpai_u}. Then we can directly obtain $\overline{\Psi}_\mathrm{A}^{\mu}\left(P_\mathrm{th}\right)>\varepsilon$ if and only if condition \eqref{mu_BS_density_1} is satisfied, which completes the proof.
\end{proof}

\begin{corollary}
 The ambient RF energy harvesting in the mmWave tier  outperforms that in the sub-6 GHz tier under the condition \eqref{mu_BS_density_1}, where $\vartheta_\mathrm{mm}$ is given by \eqref{varpai_u} with $\varepsilon=$\\$\min\left\{{\mathrm{var}\left[\Xi\right]}/{\left(\frac{P_\mathrm{th}}{\eta_\mu {P_\mu}}-\mathbb{E}\left[\Xi\right]\right)^2},1\right\}$.
\end{corollary}

It is indicated from \textbf{Corollary 3} that in practice, it is still possible that the amount of ambient RF energy harvested from the mmWave tier is larger than that from the sub-6 GHz tier.

\subsection{Power Transfer Mode Selection}
In the above, we have analyzed and compared the wireless energy harvesting in the sub-6 GHz and mmWave tiers. Here, we consider mode selection for WPT in hybrid 5G scenario, i.e., UEs are at liberty to select a sub-6 GHz BS or mmWave BS for maximizing the directed \nolinebreak[4] transferred power, since the amount of harvested energy from ambient RF is much smaller compared to that from directed power transfer,~see \cite{Kaibin_Huang2015,Lifeng_wang2015_globcom,Robert_heath_mmWave_energy}, which is also illustrated in the simulation results of section V. Thus, we have the following Propositions.

\begin{proposition}
The association probability that a UE selects the sub-6 GHz WPT is given by
\begin{align}\label{prob_mu_wave}
\mathcal{H}_\mu&=2 \pi \lambda_\mu \int_0^\infty r  \exp\Big(-2 \pi \lambda_\mathrm{mm} \Big( \int_0^{\hat{R}^{\rm{mm}}_\mathrm{LoS}\left(r\right)}f_{\Pr}(t) t dt \nonumber\\
&+\int_0^{\hat{R}^{\rm{mm}}_\mathrm{NLoS}\left(r\right)} \left(1-f_{\Pr}(t)\right) t dt \Big)-\pi \lambda_\mu r^2\Big)  dr,
\end{align}
where $\hat{R}^{\rm{mm}}_\mathrm{LoS}\left(r\right)=\left( \varpi \frac{\beta^{\rm{mm}}_\mathrm{LoS}}{\beta_\mu}\right)^{{1}/{\alpha^{\rm{mm}}_\mathrm{LoS}}} r^{\frac{\alpha_\mu}{\alpha^{\rm{mm}}_\mathrm{LoS}}} $ and $\hat{R}^{\rm{mm}}_\mathrm{NLoS}\left(r\right)= \left( \varpi \frac{\beta^{\rm{mm}}_\mathrm{NLoS}}{\beta_\mu}\right)^{{1}/{\alpha^{\rm{mm}}_\mathrm{NLoS}}} r^{\frac{\alpha_\mu}{\alpha^{\rm{mm}}_\mathrm{NLoS}}}$  with $\varpi=\frac{\eta_\mathrm{mm}{{{P_\mathrm{mm}}}} \mathrm{M_B}\mathrm{M_D}}{\eta_\mu{{{P_\mu}}} N}$.
\end{proposition}

\begin{proof}
We note that in the mmWave cell, the small-scale fading is negligible, and the directed transferred power is dominated by the mmWave pathloss. In the sub-6 GHz cell,  the small-scale fading is averaged out when the number of BS antennas is large, i.e., $\left\|\mathbf{h}_o\right\|^2 \approx N$. Therefore, the probability that a UE selects the sub-6 GHz WPT can be expressed as
\begin{align}\label{LM_1}
&\mathcal{H}_\mu =\Pr\left(\Lambda_{\mu_1}>\Lambda_{\mathrm{mm}_1}\right) \nonumber\\
&=\mathbb{E}_{\left| {{X_{{o}}}} \right|}\left\{\Pr\left(  \eta_\mu{{{P_\mu}}} N L\left( {\left| {{X_{{o}}}} \right|} \right)> \eta_\mathrm{mm}{{{P_\mathrm{mm}}}} \mathrm{M_B}\mathrm{M_D} L\left( R \right)
 \right)  \right\} \nonumber\\
&\mathop=\limits^{(e)} \mathbb{E}_{\left| {{X_{{o}}}} \right|}\Bigg[ \underbrace{\Pr\Big(R_\ell^{{\rm{mm}}}> \Big( \varpi \frac{\beta^{\rm{mm}}_\mathrm{LoS}}{\beta_\mu}\Big)^{{1}/{\alpha^{\rm{mm}}_\mathrm{LoS}}} \left| {{X_{{o}}}} \right|^{\frac{\alpha_\mu}{\alpha^{\rm{mm}}_\mathrm{LoS}}} \left| {\ell \in \Phi_\mathrm{mm}^{\mathrm{LoS}}} \right. \Big)}_{\Theta_\mathrm{L}\left(\left| {{X_{{o}}}} \right|\right)}\nonumber\\
&\times \underbrace{\Pr\Big(R_\ell^{{\rm{mm}}}> \Big( \varpi \frac{\beta^{\rm{mm}}_\mathrm{NLoS}}{\beta_\mu}\Big)^{{1}/{\alpha^{\rm{mm}}_\mathrm{NLoS}}} \left| {{X_{{o}}}} \right|^{\frac{\alpha_\mu}{\alpha^{\rm{mm}}_\mathrm{NLoS}}} \left| {\ell \in \Phi_\mathrm{mm}^{\mathrm{NLoS}}} \right. \Big)}_{_{\Theta_\mathrm{N}\left(\left| {{X_{{o}}}} \right|\right)}}
\Bigg] \nonumber\\
&=\int_0^\infty {\Theta_\mathrm{L}}\left(r\right) {\Theta_\mathrm{N}}\left(r\right) f_{\left| {{X_{{o}}}} \right|} \left( r \right) dr,
\end{align}
where step (e) is obtained by  considering two independent LoS BS process $\Phi_\mathrm{mm}^{\mathrm{LoS}}$ and NLoS BS process $\Phi_\mathrm{mm}^{\mathrm{NLoS}}$, $\varpi=\frac{\eta_\mathrm{mm}{{{P_\mathrm{mm}}}} \mathrm{M_B}\mathrm{M_D}}{\eta_\mu{{{P_\mu}}} N}$, $f_{\left| {{X_{{o}}}} \right|} \left( r \right)$ is the PDF of $\left| {{X_{{o}}}} \right|$ given in \eqref{pf1_2}. By employing the void probability, we can obtain $\Theta_\mathrm{L}$ as
\begin{align}\label{theta_LoS_BS}
\Theta_\mathrm{L}\left(r\right)=&\Pr\Big(\mathrm{No~LoS~mmWave~BS~closer~than} ~~\hat{R}^{\rm{mm}}_\mathrm{LoS}\left(r\right)\Big)\nonumber\\
=&\exp\Big(-2\pi\lambda_\mathrm{mm}\int_0^{\hat{R}^{\rm{mm}}_\mathrm{LoS}\left(r\right)} f_{\Pr}(t) t dt \Big),
\end{align}
where $\hat{R}^{\rm{mm}}_\mathrm{LoS}\left(r\right)=\left( \varpi {\beta^{\rm{mm}}_\mathrm{LoS}}/{\beta_\mu}\right)^{{1}/{\alpha^{\rm{mm}}_\mathrm{LoS}}} r^{{\alpha_\mu}/{\alpha^{\rm{mm}}_\mathrm{LoS}}} $. Similar to \eqref{theta_LoS_BS}, $\Theta_\mathrm{N}$ is given by
\begin{align}\label{theta_NLoS_BS}
\hspace{-0.3 cm}\Theta_\mathrm{N}\left(r\right)=&\exp\Big(-2\pi\lambda_\mathrm{mm}\int_0^{\hat{R}^{\rm{mm}}_\mathrm{NLoS}\left(r\right)} \left(1-f_{\Pr}(t)\right) t dt \Big),
\end{align}
where $\hat{R}^{\rm{mm}}_\mathrm{NLoS}\left(r\right)= \left( \varpi \frac{\beta^{\rm{mm}}_\mathrm{NLoS}}{\beta_\mu}\right)^{{1}/{\alpha^{\rm{mm}}_\mathrm{NLoS}}} r^{\frac{\alpha_\mu}{\alpha^{\rm{mm}}_\mathrm{NLoS}}} $. Substituting \eqref{theta_LoS_BS} and \eqref{theta_NLoS_BS} into \eqref{LM_1}, we obtain the desired result \eqref{prob_mu_wave}.
\end{proof}
\begin{proposition}
The association probability that a UE selects a LoS mmWave BS for the mmWave WPT is given by
\begin{align}\label{prob_mm_wave}
\mathcal{H}_{\mathrm{mm}}^{\mathrm{LoS}}&=2 \pi \lambda_\mathrm{mm} \int_0^\infty r f_{\Pr}(r)  \exp\Big(-2 \pi \lambda_\mathrm{mm} \Big( \int_0^{r}f_{\Pr}(t) t dt \nonumber\\
&+\int_0^{\widetilde{R}^{\rm{mm}}_\mathrm{NLoS}\left(r\right)} \left(1-f_{\Pr}(t)\right) t dt \Big)- \lambda_\mu \widetilde{A}_\mu\left(r\right)   \Big)  dr,
\end{align}
where ${\widetilde{R}^{\rm{mm}}_\mathrm{NLoS}\left(r\right)}=\left(\frac{\beta_\mathrm{NLoS}^{\mathrm{mm}}}{\beta_\mathrm{LoS}  ^{\mathrm{mm}}}\right)^{1/\alpha^{\rm{mm}}_\mathrm{NLoS}} r^{\frac{\alpha^{\rm{mm}}_\mathrm{LoS}}{\alpha^{\rm{mm}}_\mathrm{NLoS}}   }$ and $\widetilde{A}_\mu\left(r\right)=\pi{\left(\frac{ \beta_\mu}{\varpi \beta^{\rm{mm}}_\mathrm{LoS}} \right)^{2/\alpha_\mu} r^{ \frac{2\alpha^{\rm{mm}}_\mathrm{LoS}}{\alpha_\mu} }  }$. Then the probability that a UE selects a NLoS mmWave BS for the mmWave WPT is $\mathcal{H}_{\mathrm{mm}}^{\mathrm{NLoS}}=1-\mathcal{H}_\mu-\mathcal{H}_{\mathrm{mm}}^{\mathrm{LoS}}$.
\end{proposition}
\begin{proof}
We first define $\tau_\mathrm{L}$ as the probability that there exist LoS mmWave BSs. Similar to \eqref{LM_1}, the probability that a UE selects a LoS mmWave BS for the mmWave energy harvesting is calculated as
\begin{align}\label{LoS_mmWave_Pr}
&\mathcal{H}_{\mathrm{mm}}^{\mathrm{LoS}}=\tau_\mathrm{L} \mathbb{E}_{\left| {{Y_{{o}}}} \right|}
\bigg\{\underbrace{\Pr\left( \Lambda_{\mathrm{mm}_1}^{\mathrm{LoS}}> \Lambda_{\mu_1} \right)}_{\Theta_\mu\left(\left| {{Y_{{o}}}} \right|\right)} \times \nonumber\\
&\underbrace{\Pr\left(\Lambda_{\mathrm{mm}_1}^{\mathrm{LoS}} > \eta_\mathrm{mm}{{{P_\mathrm{mm}}}} \mathrm{M_B}\mathrm{M_D} \beta^{\rm{mm}}_\mathrm{NLoS} ({R_\ell^{\mathrm{mm}}})^{-\alpha^{\rm{mm}}_\mathrm{NLoS}}\left| {\ell \in \Phi_\mathrm{mm}^{\mathrm{NLoS}}} \right.\right)}_{{{\overline{\Theta}}_\mathrm{N}\left(\left| {{Y_{{o}}}} \right|\right)}} \bigg\}\nonumber\\
&=\tau_\mathrm{L} \int_0^\infty \Theta_\mu\left(r\right)  {{\overline{\Theta}}_\mathrm{N}\left(r\right)} f_{\left| {{Y_{{o}}}} \right|} \left( r \right) dr,
\end{align}
where $\Lambda_{\mathrm{mm}_1}^{\mathrm{LoS}}= \eta_\mathrm{mm}{{{P_\mathrm{mm}}}} \mathrm{M_B}\mathrm{M_D} \beta^{\rm{mm}}_\mathrm{LoS} \left| {{Y_{{o}}}} \right|^{-\alpha^{\rm{mm}}_\mathrm{LoS}} $ is the directed transferred power from the nearest LoS mmWave BS, and the PDF of $\left| {{Y_{{o}}}} \right|$ is given by~\cite{Tianyang_arxiv2014}
\begin{align}\label{PDF_Y_mmWave}
f_{\left| {{Y_{{o}}}} \right|} \left( r \right)=\frac{2 \pi \lambda_\mathrm{mm}}{\tau_\mathrm{L}} r f_{\Pr}\left(r\right) e^{-2 \pi \lambda_\mathrm{mm} \int_0^r
f_{\Pr}\left(t\right) t dt }.
\end{align}
Similar to \eqref{theta_LoS_BS}, ${\Theta_\mu\left(r\right)}$ is derived as
\begin{align}\label{theta_mu_Y}
& {\Theta_\mu\left(r\right)}= \Pr\left(\eta_\mathrm{mm}{{{P_\mathrm{mm}}}} \mathrm{M_B}\mathrm{M_D} \beta^{\rm{mm}}_\mathrm{LoS} {r}^{-\alpha^{\rm{mm}}_\mathrm{LoS}}> \eta_\mu{{{P_\mu}}} N \beta_\mu R_\mu^{-\alpha_\mu} \right) \nonumber\\
&=\Pr\Big(R_\mu > {\left(\frac{ \beta_\mu}{\varpi \beta^{\rm{mm}}_\mathrm{LoS}} \right)^{1/\alpha_\mu} r^{ \frac{\alpha^{\rm{mm}}_\mathrm{LoS}}{\alpha_\mu} }  }      \Big)\nonumber\\
&=\exp\left(- \lambda_\mu \widetilde{A}_\mu\left(r\right) \right),
\end{align}
where $\widetilde{A}_\mu\left(r\right)=\pi{\left(\frac{ \beta_\mu}{\varpi \beta^{\rm{mm}}_\mathrm{LoS}} \right)^{2/\alpha_\mu} r^{ \frac{2\alpha^{\rm{mm}}_\mathrm{LoS}}{\alpha_\mu} }  }$.
Then ${{\overline{\Theta}}_\mathrm{N}\left(r\right)}$ is similarly derived as
\begin{align}\label{theta_N_NLOS_mmwave}
\hspace{-0.2 cm}{{\overline{\Theta}}_\mathrm{N}\left(r\right)}=\exp\Big(-2\pi\lambda_\mathrm{mm}\int_0^{\widetilde{R}^{\rm{mm}}_\mathrm{NLoS}\left(r\right)} \left(1-f_{\Pr}(t)\right) t dt \Big),
\end{align}
where ${\widetilde{R}^{\rm{mm}}_\mathrm{NLoS}\left(r\right)}=\left(\frac{\beta_\mathrm{NLoS}^{\mathrm{mm}}}{\beta_\mathrm{LoS}  ^{\mathrm{mm}}}\right)^{1/\alpha^{\rm{mm}}_\mathrm{NLoS}} r^{\frac{\alpha^{\rm{mm}}_\mathrm{LoS}}{\alpha^{\rm{mm}}_\mathrm{NLoS}}   }$.  Substituting \eqref{theta_mu_Y} and \eqref{theta_N_NLOS_mmwave} into \eqref{LoS_mmWave_Pr}, we obtain \eqref{prob_mm_wave} and complete the proof.
\end{proof}

When only LoS mmWave links with $f_\mathrm{Pr}\left(R\right)=\mathbf{1}\left(R \leq R_\mathrm{LoS}\right)$ are able to transfer energy by mmWave radiation, the association probability that a UE selects the sub-6 GHz WPT given by \eqref{prob_mu_wave} reduces to
\begin{align}\label{prob_mu_wave_simplified}
\widetilde{\mathcal{H}}_\mu&=2 \pi \lambda_\mu \int_0^\infty r  \exp\Big(- \pi \lambda_\mathrm{mm}\left(\hat{\hat{R}}\left(r\right)\right)^2-\pi \lambda_\mu r^2\Big)  dr,
\end{align}
where $\hat{\hat{R}}\left(r\right)=\min\bigg( \left(\mathrm{M_B}\mathrm{M_D}\right)^{{1}/{\alpha^{\rm{mm}}_\mathrm{LoS}}} \left(  \frac{\eta_\mathrm{mm}{{{P_\mathrm{mm}}}} \beta^{\rm{mm}}_\mathrm{LoS}}{\eta_\mu{{{P_\mu}}} N \beta_\mu }\right)^{{1}/{\alpha^{\rm{mm}}_\mathrm{LoS}}}$\\$ r^{\frac{\alpha_\mu}{\alpha^{\rm{mm}}_\mathrm{LoS}}},R_\mathrm{LoS}\bigg)$.
Accordingly, the association probability that a UE selects a LoS mmWave BS for the mmWave WPT is $\widetilde{\mathcal{H}}_{\mathrm{mm}}^{\mathrm{LoS}}=1-\widetilde{\mathcal{H}}_\mu$. In light of $\alpha^{\rm{mm}}_\mathrm{LoS}\geq 2$~\cite{S_A_H_WPT_2015}, we find that compared to increasing the mmWave antenna gain $\mathrm{M_B}\mathrm{M_D}$, $\widetilde{\mathcal{H}}_{\mathrm{mm}}^{\mathrm{LoS}}$ grows at a higher speed when increasing the mmWave density $\lambda_\mathrm{mm}$, which highlights the importance of achieving more mmWave BS densification gains for mmWave WPT.

\section{Throughput Analysis}
The previous section has revealed how many small cells are required for obtaining the expected harvested energy. In this section, we characterize the uplink performance in terms of the throughput in the considered networks. We will calculate how many small cells should be deployed in the sub-6 GHz and mmWave tiers, in order to achieve the targeted throughput. Moreover, we will answer how many mmWave cells are needed to outperform a specific sub-6 GHz tier. We assume that in each block time $T$, UEs first harvest RF energy with the time duration $\tau T$ ($0<\tau<1$), and the rest of time is allocated for uplink transmissions
by fully using the harvested energy.

\subsection{Sub-6 GHz Tier} In the sub-6 GHz tier, we  define the ratio between the sub-6 GHz BS density and active $\mu$UE density as $\kappa_{\mu\mathrm{UE}}={\lambda_\mu}/{ \widetilde{\lambda}_{\mu\mathrm{UE}}}$, to characterize the tier scale. For a given active $\mu$UE density, larger $\kappa_{\mu\mathrm{UE}}$ means that the network is denser and the distance between a $\mu$UE and its serving BS is shorter, which brings more BS densification gains. Note that the achievable BS densification gain can help combat the uplink interference without using complicated interference coordination methods such as \cite{WeiFeng_2013_JSAC}. For a given sub-6 GHz BS density, lower $\kappa_{\mu\mathrm{UE}}$ means that more active $\mu$UEs are served, which results in larger uplink interference.

 We first derive the exact expression for the throughput $\mathrm{C}_\mu$ between a typical $\mu$UE and its serving sub-6 GHz BS as
\begin{align}\label{muWave_rate_exact}
\mathrm{C}_\mu=\frac{1 - \tau }{\ln2} \mathrm{BW}_\mu \int_0^\infty \frac{\varphi_\mu(t)}{t} e^{-{\widetilde{\lambda}_{\mu\mathrm{UE}}}^{\frac{-\alpha_\mu}{2}}{\sigma^2}t}  dt,
\end{align}
where  $\mathrm{BW}_\mu$ is the sub-6 GHz bandwidth, $\varphi_\mu(t)$ is given by  \eqref{varphi_t1_1} (see top of this page), in which
$\hbar_\mu=\frac{\tau}{1-\tau} \eta_\mu{{{P_\mu}}} N {\lambda_\mu^{\frac{\alpha_\mu}{2}}}$.
\begin{figure*}[htbp]
\normalsize
\begin{align}\label{varphi_t1_1}
&\varphi_\mu\left(t\right)=\frac{2\pi}{e^{-\pi r_o^2}} \widetilde{I}_\mu(t) \int_{r_o}^\infty  \Big( 1 -  {{e^{ - t\kappa _{\mu {\rm{UE}}}^{\frac{\alpha_\mu}{2}} N \hbar_\mu L^2(r)}}}  \Big)  r{e^{ - \pi {r^2}}}dr
\end{align}
with
\begin{align}\label{I_tilde_112}
\widetilde{I}_\mu(t)=\exp\Big\{-\frac{4 \pi^2}{{e^{-\pi r_o^2}}} \int_{r_o}^\infty \int_{r_o}^\infty \frac{t\hbar_\mu L(y) L\left( x \right)} {{1{\rm{ + }}t\hbar_\mu L(y) L\left( x \right)} } xy e^{-\pi y^2}dy dx\Big\}
\end{align}
\hrulefill
\end{figure*}
\begin{proof}
See Appendix A.
\end{proof}

To further shed light on the effect of $\kappa_{\mu\mathrm{UE}}$ on the throughput, we provide a closed-form lower bound expression for \eqref{muWave_rate_exact}, which is given by
\begin{align}\label{muwave_BS}
& \mathrm{C}_\mu^{\mathrm{L}}=(1 - \tau)\mathrm{BW}_\mu \times \nonumber\\
&~~~\log_2\left(1+ \kappa_{\mu\mathrm{UE}}^{\frac{\alpha_\mu}{2}} \frac{ N e^{\alpha_\mu {e^{\pi r_o^2}} \zeta_o}}{2\pi^2 \zeta_1 {e^{\pi r_o^2}} r_o^{2-\alpha_\mu}E_{(\frac{\alpha_\mu}{2})}(r_o^2\pi) }\right),
\end{align}
where $\zeta_o=Ei\left(-r_o^2\pi\right)-2e^{-r_o^2\pi}\ln r_o$, $\zeta_1 =\frac{r_o^{2-\alpha_\mu}}{\alpha_\mu-2}$, and $Ei(\cdot)$ is the exponential integral function~\cite[(8.211)]{gradshteyn}.

\begin{proof}
See Appendix B.
\end{proof}

It is explicitly shown from \eqref{muwave_BS} that the throughput increases with $\kappa_{\mu\mathrm{UE}}$ and $N$, which means that adding more small cells or antennas enhances the performance. Also,
for large $\kappa_{\mu\mathrm{UE}}$ or $N$, the throughput scales
as $(1 - \tau)\mathrm{BW}_\mu \left({\frac{\alpha_\mu}{2}} \log_2\left(\kappa_{\mu\mathrm{UE}}\right)+\log_2\left(N\right)\right)$,
which indicates that the throughput grows at a higher speed when increasing $\kappa_{\mu\mathrm{UE}}$,
compared to increasing the number of BS antennas\footnote{Note that the pathloss exponent $\alpha_\mu\geq 2$ in the practical environments~\cite{S_A_H_WPT_2015}.}. 

\begin{corollary}
An expected throughput $C_{\mathrm{th}}^\mu$ is achievable when the ratio between the sub-6 GHz BS density and active $\mu$UE density
satisfies
\begin{align}\label{kappa_123}
\hspace{-0.2 cm}\kappa_{\mu\mathrm{UE}} \geq N^{-\frac{2}{\alpha_\mu}}   \Big( \frac{\big({{\rm{2}}^{\frac{{\rm{C}}_{\mathrm{th}}^\mu}{\mathrm{BW}_\mu(1 - \tau)}}} - 1\big)2\pi^2 \zeta_1 {e^{\pi r_o^2}} r_o^{2-\alpha_\mu}E_{(\frac{\alpha_\mu}{2})}(r_o^2\pi)} {e^{\alpha_\mu {e^{\pi r_o^2}} \zeta_o}}\Big)^{\frac{2}{\alpha_\mu}}.
\end{align}
\end{corollary}

From \textbf{Corollary 4},  the ratio between the sub-6 GHz BS density and the active $\mu$UE density should be larger than a critical value for obtaining the desired performance.

\subsection{MmWave Tier} In the mmWave tier, we similarly define the ratio between the mmWave BS density and the active mmUE density as $\kappa_{\mathrm{mmUE}}={\lambda_\mathrm{mm}}/{ \widetilde{\lambda}_{\mathrm{mmUE}}}$. The throughput can be derived as
\begin{align}\label{mmWave_rate_2016}
C_\mathrm{mm}=\frac{\left(1-\tau\right)}{\ln2} \mathrm{BW}_\mathrm{mm} \int_0^\infty \frac{\varphi_{\mathrm{mm}}(t)}{t} e^{-{\sigma^2}t}  dt,
\end{align}
where $\mathrm{BW}_\mathrm{mm}$ denotes the mmWave bandwidth, $\hbar_\mathrm{mm}=\frac{\tau}{1-\tau}\eta_\mathrm{mm}{{{P_\mathrm{mm}}}} \mathrm{M_\mathrm{B}}\mathrm{M_\mathrm{D}} $, and $\varphi_{\mathrm{mm}}(t)$ is given by \eqref{varphi_t1_mmWave1} at the next page.

\begin{figure*} [t!]
\normalsize
\begin{align}\label{varphi_t1_mmWave1}
&\varphi_\mathrm{mm}\left(t\right)=\frac{2\pi \lambda_\mathrm{mm} }{e^{-\pi \lambda_\mathrm{mm} r_o^2}}\widetilde{I}_\mathrm{mm}(t) \int_{r_o}^{R_\mathrm{LoS}} \Big( 1 -  {{e^{ - t \hbar_\mathrm{mm} \mathrm{M}_D \mathrm{M}_B (\beta^{\rm{mm}}_\mathrm{LoS})^2 r^{-2\alpha^{\rm{mm}}_\mathrm{LoS}}}}}\Big)    r \exp\left(-\pi \lambda_\mathrm{mm} r^2 \right) dr,
\end{align}
where
\begin{align}\label{mmWave_interfernce_average}
\widetilde{I}_\mathrm{mm}(t)=\exp \left\{ { - 2\pi  {\lambda_\mathrm{mm}} \kappa_{\mathrm{mmUE}}^{-1} \int_{r_o}^{R_\mathrm{LoS}} {\left(1-\phi\left(y\right)\right) y} dy} \right\}
\end{align}
with
\begin{align*}
\phi\left(y\right)=\frac{2\pi \lambda_\mathrm{mm} }{e^{-\pi \lambda_\mathrm{mm} r_o^2}} \sum\limits_{\ell  \in \left\{ {{\mathrm{M_B}},{\mathrm{m_B}}} \right\}} \sum\limits_{k \in \left\{ {{\mathrm{M_D}},{\mathrm{m_D}}} \right\}} \mathrm{Pr}_{\ell k}  \int_{r_o}^{R_\mathrm{LoS}} e^{-t\hbar_\mathrm{mm} \ell k (\beta^{\rm{mm}}_\mathrm{LoS})^2  (zy)^{-2\alpha^{\rm{mm}}_\mathrm{LoS}} }
z \exp\left(-\pi \lambda_\mathrm{mm} z^2 \right) dz
\end{align*}
\hrulefill
\end{figure*}
\begin{proof}
See Appendix C.
\end{proof}

 In the wireless powered mmWave tier, the transmit power of a mmUE is much lower due to the limited harvested energy resulting from the energy loss of propagation and the RF-to-DC conversion, which means that it is more likely to be noise-limited in the wireless powered uplink mmWave tier. As such, it is necessary to analyze the throughput in the noise-limited mmWave scenario, where uplink interference is negligible. Note that such analysis can also be viewed as a tight upper bound of the exact throughput given by \eqref{mmWave_rate_2016}, and has a good approximation to \eqref{mmWave_rate_2016} in the practical urban settings~\cite{Singh_2015}. Therefore, the throughput expression in \eqref{mmWave_rate_2016} can be further simplified  as
 \begin{align}\label{snr_mmwave_rate}
& C_\mathrm{mm}^{\rm n}={\left(1-\tau\right)} \mathrm{BW}_\mathrm{mm}  \times \nonumber\\
&\int_{r_o}^{R_\mathrm{LoS}} \log_2\left(1+\hbar_\mathrm{mm}(\mathrm{M}_D \mathrm{M}_B \beta _{{\rm{LoS}}}^{{\rm{mm}}})^2\frac{r^{-2 \alpha _{{\rm{LoS}}}^{{\rm{mm}}}}}{
{\sigma^2}}\right) \widetilde{f}_{\left| {{Y}} \right|} \left( r \right) dr,
 \end{align}
 where $\widetilde{f}_{\left| {{Y}} \right|} \left( \cdot \right)=\frac{2\pi \lambda_\mathrm{mm} r}{e^{-\pi \lambda_\mathrm{mm} r_o^2}}\exp\left(-\pi \lambda_\mathrm{mm} r^2 \right)$  is the modified PDF of the distance $\left| {{Y}} \right|$ between a mmUE and its nearest mmWave BS under the constraint $\left| {{Y}} \right| \geq r_o$.

 We next derive a closed-form lower bound expression for \eqref{snr_mmwave_rate} as
 \begin{align}\label{NL_mmWave1}
 C_\mathrm{mm}^{\rm L}={\left(1-\tau\right)} \mathrm{BW}_\mathrm{mm}  \log_2\left(1+\frac{e^{\widetilde{\varphi}_\mathrm{mm}(\lambda_\mathrm{mm})} }{
{\sigma^2}}\right).
 \end{align}
In \eqref{NL_mmWave1}, ${\widetilde{\varphi}_\mathrm{mm}(\lambda_\mathrm{mm})}$ is an increasing function of $\lambda_\mathrm{mm}$, which is given by
 \begin{align}\label{varphi_123}
{\widetilde{\varphi}_\mathrm{mm}(\lambda_\mathrm{mm})}=& \Big(\ln(\hbar_\mathrm{mm}\mathrm{M}_D \mathrm{M}_B)+2\big( {\ln \beta _{{\rm{LoS}}}^{{\rm{mm}}} + \frac{{\alpha _{{\rm{LoS}}}^{{\rm{mm}}}}}{2}\ln \big( {\pi \lambda_\mathrm{mm} } \big)} \big)\Big) \nonumber\\
&\times \Big( {{e^{ - {a_1}}} - {e^{ - {b_1}}}} \Big)- \alpha _{{\rm{LoS}}}^{{\rm{mm}}}\Big( Ei\left( { - {b_1}} \right) \nonumber\\
&+\Gamma \left( {0,{a_1}} \right) + {e^{ - a_1}\ln \left( {{a_1}} \right) - {e^{ - {b_1}}}\ln {b_1}} \Big)
 \end{align}
with $a_1=\pi {\lambda _{\rm mm}}r_o^2$ and $b_1=\pi {\lambda _{\rm mm}}R_{{\rm{LoS}}}^2$.

\begin{proof}
See Appendix D.
\end{proof}

\begin{corollary}
Based on \eqref{NL_mmWave1}, we find that the expected throughput $C_{\mathrm{th}}^\mathrm{mm}$ can be achieved when mmWave density $\lambda_\mathrm{mm}$ satisfies the following  condition:
\begin{align}\label{lambda_mmWave_123}
\lambda_\mathrm{mm} \geq  {\widetilde{\varphi}_\mathrm{mm}^{-1}\left(\ln \sigma^2\left( {2^{\frac{C_{\mathrm{th}}^\mathrm{mm}}{{\left(1-\tau\right)} \mathrm{BW}_\mathrm{mm}}}-1}\right) \right)},
\end{align}
where ${\widetilde{\varphi}_\mathrm{mm}^{-1}}$ is the inverse function of $\widetilde{\varphi}_\mathrm{mm}$.
\end{corollary}
It is confirmed by \textbf{{Corollary~5}} that the mmWave density needs to be larger than the RHS of \eqref{lambda_mmWave_123}, in order to achieve a targeted throughput.

\begin{corollary}
Based on {\bf Corollary~4}, a $\mu$$\mathrm{UE}$ can achieve a higher throughput than a $\rm{mmUE}$ when $\kappa_{\mu\mathrm{UE}}$ satisfies \eqref{kappa_123} with $C_{\mathrm{th}}^\mu=C_\mathrm{mm}$ given in \eqref{mmWave_rate_2016}. Based on {\bf Corollary~5}, a $\mathrm{mmUE}$ can achieve a higher throughput than a $\mu$$\mathrm{UE}$ when  $\lambda_\mathrm{mm}$ satisfies \eqref{lambda_mmWave_123} with $C_{\mathrm{th}}^\mathrm{mm}=\mathrm{C}_\mu$ given in~\eqref{muWave_rate_exact}.
\end{corollary}

\textbf{{Corollary~6}} confirms that whether the sub-6 GHz tier performs better than the mmWave tier depends on the number of sub-6 GHz small cells in the 5G networks.

\section{Simulation Results}
In this section, we present simulation results to show the energy coverage and throughput performance in wireless powered 5G dense cellular networks. The results validate the prior analysis, and further illustrate the impacts of node density on the RF energy harvesting and information transmission. The basic simulation parameters are shown in Table~\ref{table:1}.
\begin{table}
\newcommand{\tabincell}[2]{\begin{tabular}{@{}#1@{}}#2\end{tabular}}\scriptsize
\centering
\caption{Simulation Parameters}
 \begin{tabular}{|l|c|}
\hline
\textbf{Parameter} & \textbf{Value}   \\
\hline
 MmWave carrier frequency & $f_\mathrm{mm}=28$ GHz \\
\hline
RF-to-DC conversion efficiency & $\eta_\mu=\eta_\mathrm{mm}=0.6$\\
\hline
BS transmit power & $P_\mu=P_\mathrm{mm}=30$ dBm \\
\hline
Reference distance & $r_o=1$ \\
\hline
MmWave pathloss exponent & $\alpha^{\rm{mm}}_\mathrm{LoS}=2$, $\alpha^{\rm{mm}}_\mathrm{NLoS}=2.9$\\
\hline
\end{tabular}
\label{table:1}
\end{table}

\subsection{Energy Coverage}
In this subsection, we study energy coverage in the sub-6 GHz and mmWave tiers. It is assumed that the LoS probability function is $f_\mathrm{Pr}\left(R\right)=e^{-\varrho R}$ with $1/{\varrho}= 141.4$ m~\cite{Tianyang_arxiv2014}, the sub-6 GHz carrier frequency is  $f_c=1.5$ GHz, and the mmWave antenna beam pattern at a mmWave BS and mmUE are $(\mathrm{M_B}, \mathrm{m_B}, \theta_\mathrm{B})=(18~\mathrm{dB}, -2~\mathrm{dB}, 10^{o})$ and $(\mathrm{M_D}, \mathrm{m_D}, \theta_\mathrm{D})=(10~\mathrm{dB}, -10~\mathrm{dB}, 45^{o})$, respectively.

Fig.~\ref{fig1} shows the directed transferred energy coverage probability results in the sub-6 GHz and mmWave tiers. The analytical results are obtained from \eqref{mu_BS_approx} and \eqref{mmWave_pf1}, respectively, which are validated by Monte Carlo simulations. The result in \eqref{mu_BS_approx} can  predict the energy coverage of the sub-6 GHz tier. MmWave power transfer can be better than the sub-6 GHz  counterpart, due to  the mmWave directivity gain and densification gain.
\begin{figure}[htbp]
\centering
\includegraphics[width=3.5 in,height=2.9 in]{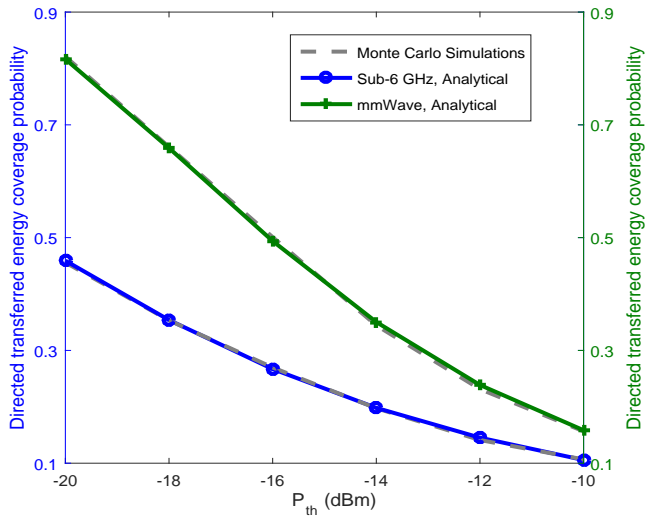}
\caption{Directed transferred energy coverage probability with $N=32$, $\alpha_\mu=2.7$, $\lambda_\mu=0.002$ and $\lambda_\mathrm{mm}=0.02$.}
\label{fig1}
\end{figure}

\begin{figure}[htbp]
\centering
\includegraphics[width=3.5 in,height=2.9 in]{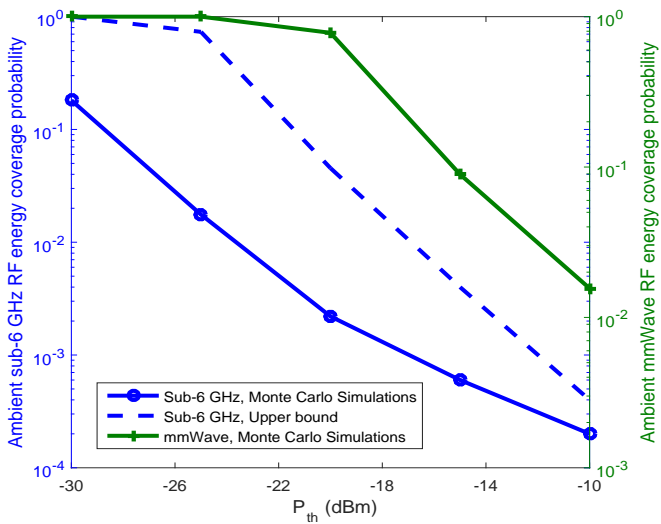}
\caption{Ambient RF energy coverage probability with $\alpha_\mu=2.6$, $\lambda_\mu=0.002$ and $\lambda_\mathrm{mm}=0.5$.}
\label{fig2}
\end{figure}

Fig.~\ref{fig2} shows the ambient RF energy coverage probability results in the sub-6 GHz and mmWave tiers. We observe that for ultra-dense mmWave tier, the harvested ambient mmWave RF energy can still be larger than that in the sub-6 GHz tier with comparably lower BS density. Compared to Fig.~\ref{fig1}, it is indicated that the amount of harvested energy from the ambient RF is much lower than the directed power transfer.

\begin{figure}[t!]
\centering
\includegraphics[width=3.5 in,height=2.9 in]{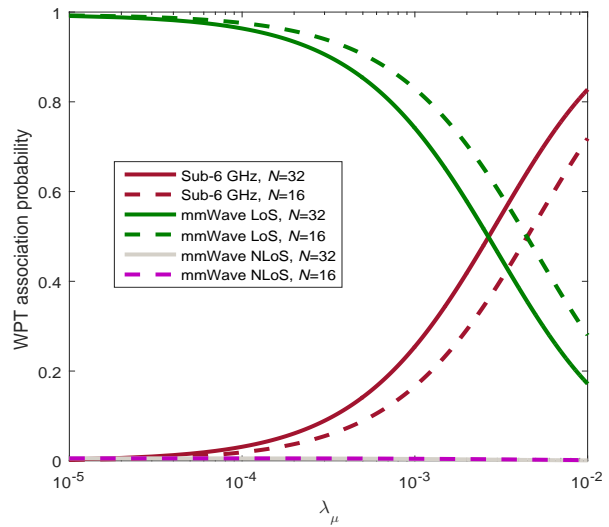}
\caption{WPT association probability with $\alpha_\mu=2.7$ and $\lambda_\mathrm{mm}=0.01$.}
\label{fig3}
\end{figure}

Fig.~\ref{fig3} demonstrates the association probability that a UE  selects a sub-6 GHz BS or mmWave BS in hybrid cellular networks. The results are obtained based on \textbf{Proposition 1} and \textbf{Proposition 2}. We observe that the association probability for sub-6 GHz WPT increases with the sub-6 GHz BS density because of obtaining higher densification gains, and it will also be improved by adding sub-6 GHz antennas for achieving more antenna gains. More UEs will select mmWave BSs to transfer energy when the sub-6 GHz BS density is much lower than the mmWave BS density, which implies that dense small cells are needed to shorten the energy transfer distance between the UE and its associated BS. Compared to the mmWave LoS, the association probability that a UE selects a mmWave NLoS power transfer is negligible.

\subsection{Throughput}
Here, we study the impact of the ratio between the BS density and the active UE density on the throughput. In the simulations, the energy harvesting time allocation factor is $\tau=0.7$, the sub-6 GHz carrier frequency is  $f_c=1$ GHz, the sub-6 GHz pathloss exponent is $\alpha_\mu=2.6$, the sub-6 GHz bandwidth is $\mathrm{BW}_\mu=20$ MHz, the mmWave bandwidth is $\mathrm{BW}_{\rm mm}=1$ GHz, and the mmWave antenna beam pattern at an active mmUE and a mmWave BS are $(\mathrm{M_D}, \mathrm{m_D}, \theta_\mathrm{D})=(3~\mathrm{dB}, -3~\mathrm{dB}, 45^{o})$ and $(\mathrm{M_B}, \mathrm{m_B}, \theta_\mathrm{B})=(18~\mathrm{dB}, -2~\mathrm{dB}, 10^{o})$, respectively, and the maximum LoS distance is $R_\mathrm{LoS}=20$ m. The noise power is obtained by  $\sigma^2 = -174 + 10\log 10$(BW)$ + \mathrm{Nf}$ dBm with $7$ dB noise figure (Nf).

\begin{figure}[htbp]
\centering
\includegraphics[width=3.5 in,height=2.9 in]{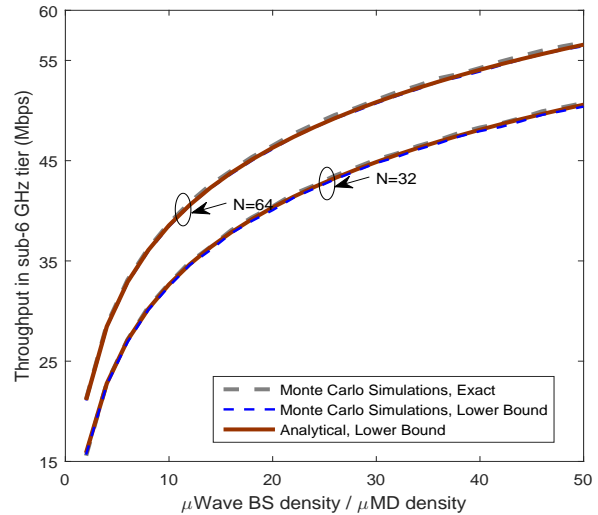}
\caption{Throughput in sub-6 GHz tier with ${ \widetilde{\lambda}_{\mu\mathrm{UE}}}=0.001$.}
\label{fig4}
\end{figure}

Fig.~\ref{fig4} shows the throughput in the sub-6 GHz tier. The analytical lower bound curves are obtained from \eqref{muwave_BS}, which tightly matches with the simulated exact curves. We see that deploying more sub-6 GHz small cells can significantly increase the throughput, due to the densification gain. Adding more BS antennas can further enhance the spectrum efficiency and bring an increase in throughput. It is also indicated from Fig.~\ref{fig4} that given a specific throughput level, the required number of small cells can be cut by using more antennas at each BS.

\begin{figure}[htbp]
\centering
\includegraphics[width=3.5 in,height=2.9 in]{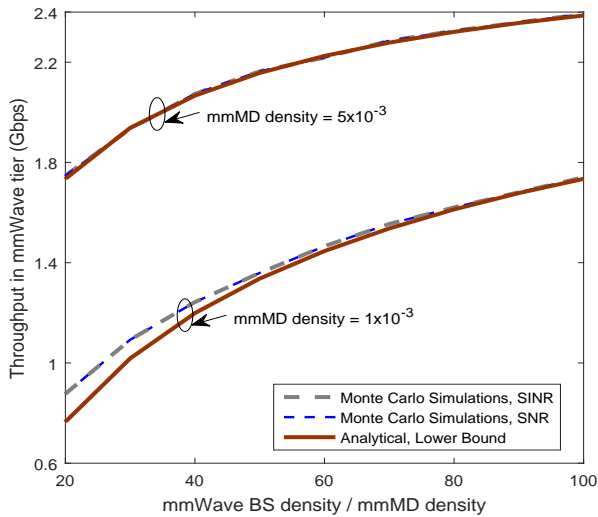}
\caption{Throughput in mmWave tier.}
\label{fig5}
\end{figure}

Fig.~\ref{fig5} illustrates the throughput in the mmWave tier. The simulated throughput curves based on the SINR has a good match with that based on SNR, confirming that the wireless powered uplink mmWave tier is noise-limited. The analytical lower bounds are obtained from \eqref{NL_mmWave1}, which can well approximate the simulated exact curves. We find that adding more sub-6 GHz small cells has a  substantial increase in throughput.

\begin{figure}[htbp]
\centering
\includegraphics[width=3.5 in,height=2.9 in]{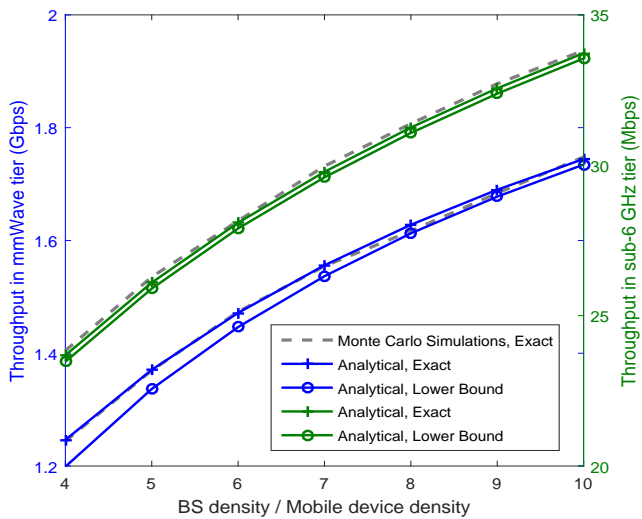}
\caption{Throughput comparison between mmWave tier and sub-6 GHz tier for the same tier scale with ${ \widetilde{\lambda}_{\mu\mathrm{UE}}}={ \widetilde{\lambda}_{\mathrm{mmUE}}}=0.01$.}
\label{fig6}
\end{figure}

Fig.~\ref{fig6} compares the throughput between the mmWave tier and the sub-6 GHz tier for the same scale, i.e., same numbers of BSs and active UEs. The analytical exact and lower bounds of the mmWave tier are obtained from \eqref{mmWave_rate_2016} and \eqref{NL_mmWave1}, respectively. The analytical exact and lower bounds of the sub-6 GHz tier are obtained from \eqref{muWave_rate_exact} and \eqref{muwave_BS}, respectively. Our analysis is validated by the simulated results. It is implied that the Gbps transmission rate is still likely to be achieved in the wireless powered dense mmWave tier, which significantly outperforms the sub-6 GHz tier. Moreover, in the wireless powered ultra-dense mmWave scenarios, interference is still negligible, i.e., noise-limited.
\begin{figure}[htbp]
\centering
\includegraphics[width=3.5 in,height=2.9 in]{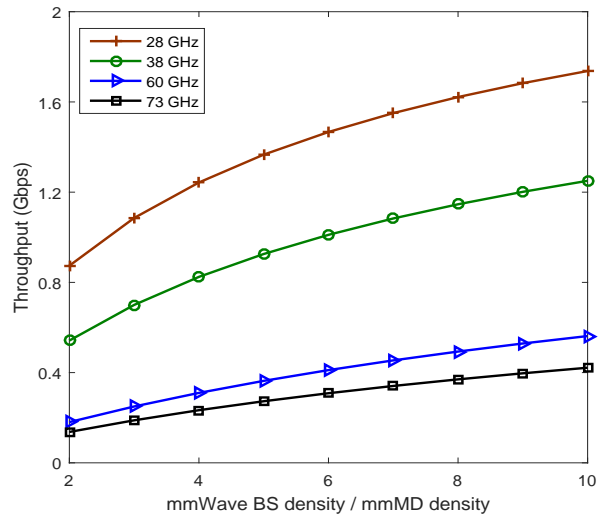}
\caption{Throughput comparison for different mmWave carrier frequencies with ${ \widetilde{\lambda}_{\mathrm{mmUE}}}=0.01$.}
\label{fig7}
\end{figure}

Fig.~\ref{fig7} provides the throughput comparison for different mmWave carrier frequencies. The LoS pathloss exponent is set as 2 at 28 GHz, 38 GHz and 73 GHz, and its value is 2.25 at 60 GHz~\cite{deng201528,rappaport201238}. We observe that the achievable throughput is the highest at 28 GHz since it has the lowest propagation loss. The performance gap between different mmWave  carrier frequencies is larger when increasing the BS density, which indicates that lower mmWave frequencies can obtain more densification gains. In addition, the performance at 73 GHz is close to that at 60 GHz, due to the fact that the atmospheric absorption at 60 GHz is more severe than that at 73 GHz, which leads to higher LoS pathloss exponent at 60 GHz.

\section{Conclusions}
This paper studied WPT in dense cellular networks, where a large number of sub-6 GHz and mmWave BSs with antenna arrays are deployed to power UEs. The expressions of the energy coverage probability were derived, in order to provide comparisons between the sub-6 GHz and mmWave energy harvesting. We obtained the BS density condition when the mmWave tier can provide more RF energy than the sub-6 GHz tier. In addition, the probability that a UE selects the sub-6 GHz or mmWave power transfer was quantified. We then derived the throughput in the uplink sub-6 GHz and mmWave tiers. We obtained the number of small cells that are required to achieve a targeted level of throughput. Also, we presented the BS density conditions when the mmWave UE achieves higher throughput than the sub-6 GHz UE.

\section*{Appendix A: A detailed derivation of \eqref{muWave_rate_exact}}
\label{App:theo_1}
\renewcommand{\theequation}{A.\arabic{equation}}
\setcounter{equation}{0}

According to the Mapping theorem~\cite{Baccelli2009}, \eqref{muWave_SINR_original} can be written as
\begin{align}\label{muWave_SINR}
\mathrm{SINR}_{\mu}&= \frac{\lambda_\mu^{\frac{\alpha_\mu}{2}} S_{\left(1 \right)}^\mu }{{\widetilde{\lambda}_{\mu\mathrm{UE}}}^{\frac{\alpha_\mu}{2}}
 I_{\left(1\right)}^\mu+{\sigma^2}}.
\end{align}
Based on \eqref{muWave_SINR}, the throughput $\mathrm{C}_\mu$ can be derived by using \cite[Lemma 1]{Hamdi_2008}, which is as follows
\begin{align}\label{muWave_exact1}
\mathrm{C}_\mu&=\frac{(1 - \tau)T  }{T}\mathrm{BW}_\mu\mathbb{E}\left[\log_2\left(1+\frac{\kappa_{\mu\mathrm{UE}}^{\frac{\alpha_\mu}{2}} S_{\left(1 \right)}^\mu }{
 I_{\left(1\right)}^\mu+{\widetilde{\lambda}_{\mu\mathrm{UE}}}^{\frac{-\alpha_\mu}{2}}{\sigma^2}}\right)\right]   \nonumber\\
&=\frac{1 - \tau}{\ln2} \mathrm{BW}_\mu \int_0^\infty \frac{\varphi_\mu(t)}{t} e^{-{\widetilde{\lambda}_{\mu\mathrm{UE}}}^{\frac{-\alpha_\mu}{2}}{\sigma^2}t}  dt,
\end{align}
where $\varphi_\mu(t)$ is
\begin{align}\label{muWave_exact1}
\varphi_\mu(t)&=\mathbb{E}\left[\left( {1 - {e^{ - t \kappa _{\mu {\rm{UE}}}^{\frac{\alpha_\mu}{2}} S_{\left( 1 \right)}^\mu }}} \right){e^{ - tI_{\left( 1 \right)}^\mu }}\right] \nonumber\\
&\approx  \int_{r_o}^\infty  \Big( 1 -  {{e^{ - t \kappa _{\mu {\rm{UE}}}^{\frac{\alpha_\mu}{2}} N P_{\mathrm{UE}_o}^{\mu} L(r)}}}  \Big) \underbrace{\mathbb{E}\left[ {{e^{ - tI_{\left( 1 \right)}^\mu }}} \right]}_{\widetilde{I}_\mu(t)} \widetilde{f}_{\left| {{X}} \right|} \left( r \right) dr.
\end{align}
In \eqref{muWave_exact1}, $\left\|\mathbf{h}_o\right\|^2 \approx N$ with large $N$, $r_o$ is the reference distance to avoid singularity at zero, and $\widetilde{f}_{\left| {{X}} \right|} \left( \cdot \right)$ is the modified PDF of the distance $\left| {X} \right|$ between a $\mu$UE and its nearest sub-6 GHz BS under the condition $\left| {{X}} \right| \geq r_o$, which is
\begin{align}\label{new_PDF}
\widetilde{f}_{\left| {{X}} \right|} \left( r \right)= \frac{2 \pi r}{e^{-\pi r_o^2}} \exp\left(-\pi r^2\right), \; r \geq r_o.
\end{align}
Note that the harvested ambient RF energy is much smaller than the directed transferred energy~\cite{Kaibin_Huang2015,Lifeng_wang2015_globcom,Robert_heath_mmWave_energy} and can be negligible, which is also illustrated in the simulation results of Section V in this paper. Therefore, based on the instantaneous harvested power given by \eqref{mu_MD_receive_power}, $P_{\mathrm{UE}_o}^{\mu}$ can be evaluated as
\begin{align}\label{varphi_2_2}
P_{\mathrm{UE}_o}^{\mu} \mathop \approx \limits^{(a)}  \hbar_\mu L(r),
\end{align}
where $\hbar_\mu=\frac{\tau}{1-\tau} \eta_\mu{{{P_\mu}}} N {\lambda_\mu^{\frac{\alpha_\mu}{2}}}$, step (a) is obtained from the Mapping theorem. Since the minimum distance between the typical BS and the interfering UEs is small in dense networks~\cite{XinqinLin_2014}, by using the Laplace functional of PPP~\cite{Baccelli2009}, $\widetilde{I}_\mu(t)$ is given by
\begin{align}\label{mmwave_exact_X1}
&\widetilde{I}_\mu(t)=\exp \left( { - 2\pi \int_{r_o}^\infty  {\mathbb{E}\left[\frac{tP_{{\rm{M}}{{\rm{D}}_i}}^\mu L\left( x \right)} {{1{\rm{ + }}tP_{{\rm{M}}{{\rm{D}}_i}}^\mu L\left( x \right)} } \right] x} dx} \right).
\end{align}
According to \eqref{varphi_2_2}, we have
\begin{align}\label{mmwave_exact_X1_1}
&\mathbb{E}\left[\frac{tP_{{\rm{M}}{{\rm{D}}_i}}^\mu L\left( x \right)} {{1{\rm{ + }}tP_{{\rm{M}}{{\rm{D}}_i}}^\mu L\left( x \right)} } \right]
=\frac{2\pi}{e^{-\pi r_o^2}} \int_{r_o}^\infty \frac{t\hbar_\mu L(y) L\left( x \right)} {{1{\rm{ + }}t\hbar_\mu L(y) L\left( x \right)} } y e^{-\pi y^2}dy.
\end{align}
Substituting \eqref{mmwave_exact_X1_1} into \eqref{mmwave_exact_X1}, we obtain  $\widetilde{I}_\mu(t)$ as \eqref{I_tilde_112} and complete the proof.
\section*{Appendix B: A detailed derivation of \eqref{muwave_BS}}
\label{App:theo_1}
\renewcommand{\theequation}{B.\arabic{equation}}
\setcounter{equation}{0}

In the dense cellular networks, the noise power is negligible. Thus, based on \eqref{muWave_SINR}, $ \mathrm{C}_\mu$ can be calculated as
\begin{align}\label{rate_mu_1}
\mathrm{C}_\mu 
 &\mathop \approx \limits^{(b)}  \frac{(1 - \tau)T  }{T} \mathbb{E}\left[\log_2\left(1+\kappa_{\mu\mathrm{UE}}^{\frac{\alpha_\mu}{2}} N \frac{ \widetilde{S}_{\left(1 \right)}^\mu }{ I_{\left(1\right)}^\mu}\right)\right],
\end{align}
where step (b) is obtained by considering $\left\|\mathbf{h}_o\right\|^2 \approx N$ with large $N$, and $\widetilde{S}_{\left(1 \right)}^\mu=P_{\mathrm{UE}_o}^{\mu}  L\left({\left| {{X_{{o}}}} \right|}\right)$.

By using Jensen's inequality~\cite{N_Lee_2016},   \eqref{rate_mu_1} can be lower bounded as
\begin{align}\label{rate_mu_2}
\mathrm{C}_\mu^{\mathrm{L}}=&(1 - \tau)\log_2\left(1+\kappa_{\mu\mathrm{UE}}^{\frac{\alpha_\mu}{2}} N \frac{e^{\mathbb{E}\left[\ln \widetilde{S}_{\left(1 \right)}^\mu  \right]}}{\mathbb{E}\left[I_{\left(1\right)}^\mu\right]}\right).
\end{align}

Considering  $P_{\mathrm{UE}_o}^{\mu} \approx  \hbar_1 L(r)$ given in \eqref{varphi_2_2}, we first calculate $\mathbb{E}\left[\ln \widetilde{S}_{\left(1 \right)}^\mu  \right]$ as
\begin{align}\label{rate_mu_3}
\mathbb{E}\left[\ln \widetilde{S}_{\left(1 \right)}^\mu  \right]&=\mathbb{E}\left[\ln P_{\mathrm{UE}_o}^{\mu}\right]+ \mathbb{E}\left[\ln  L(r) \right] \nonumber\\
&=\ln \hbar_1 +2\mathbb{E}\left[\ln  L(r) \right],
\end{align}
where $\mathbb{E}\left[\ln  L(r) \right]$ is
\begin{align}\label{rate_mu_4}
&\mathbb{E}\left[\ln  L(r) \right] =\ln \beta_\mu-\alpha_\mu  \int_{r_o}^\infty \ln(r) \widetilde{f}_{\left| {{X}} \right|} \left( r \right) dr \nonumber\\
&=\ln \beta_\mu+\frac{\alpha_\mu}{2}{e^{\pi r_o^2}}\left(Ei\left(-r_o^2\pi\right)-2e^{-r_o^2\pi}\ln r_o\right).
\end{align}
In \eqref{rate_mu_4}, $\widetilde{f}_{\left| {{X}} \right|} \left( \cdot \right)$ is given by \eqref{new_PDF}.

With the help of the Campbell's theorem, $\mathbb{E}\left[I_{\left(1\right)}^\mu\right]$ in \eqref{rate_mu_2} is calculated as
\begin{align}\label{rate_mu_8}
\mathbb{E}\left[I_{\left(1\right)}^\mu\right]&=\mathbb{E}\left[P_{\mathrm{UE}_i}^{\mu}\right] 2\pi  \beta_\mu  \int_{r_o}^\infty x^{1-\alpha_\mu}  dx \nonumber\\
& =\mathbb{E}\left[P_{\mathrm{UE}_i}^{\mu}\right] 2\pi \beta_\mu \frac{r_o^{2-\alpha_\mu}}{\alpha_\mu-2},
\end{align}
where $\mathbb{E}\left[P_{\mathrm{UE}_i}^{\mu}\right]$ can be similarly obtained based on \eqref{varphi_2_2}, which is
\begin{align}\label{rate_mu_9}
\mathbb{E}\left[P_{\mathrm{UE}_i}^{\mu}\right]&\approx \hbar_1 \beta_\mu  \int_{r_o}^\infty r^{-\alpha_\mu} \widetilde{f}_{\left| {{X}} \right|} \left( r \right)  dr \nonumber\\
&=\hbar_1 \beta_\mu \pi {e^{\pi r_o^2}} r_o^{2-\alpha_\mu}E_{(\frac{\alpha_\mu}{2})}(r_o^2\pi).
\end{align}

Substituting \eqref{rate_mu_3} and \eqref{rate_mu_8} into \eqref{rate_mu_2}, we can obtain \eqref{muwave_BS} and complete the proof.

\section*{Appendix C: A detailed derivation of \eqref{mmWave_rate_2016}}
\label{App:theo_1}
\renewcommand{\theequation}{C.\arabic{equation}}
\setcounter{equation}{0}

Similar to \eqref{muWave_exact1}, we have
\begin{align}\label{mmWave_exact1}
\mathrm{C}_\mathrm{mm} &=\frac{(1 - \tau)T  }{T} \mathrm{BW}_\mathrm{mm} \mathbb{E}\left[\log_2\left(1+\mathrm{SINR}_{\mathrm{mm}}\right)\right] \nonumber\\
&= \frac{1 - \tau}{\ln2} \mathrm{BW}_\mathrm{mm} \int_0^\infty \frac{\varphi_{\mathrm{mm}}(t)}{t} e^{-{\sigma^2}t}  dt,
\end{align}
where $\varphi_{\mathrm{mm}}(t)$ is
\begin{align}\label{mmWave_exact11}
\varphi_{\mathrm{mm}}(t)&=\mathbb{E}\Big[\Big( {1 - {e^{ - t S_{\left(\lambda_\mathrm{mm}\right)}^{\mathrm{mm}} }}} \Big){e^{ - tI_{\left(\widetilde{\lambda}_{\mathrm{mmUE}}\right)}^{\mathrm{mm}} }}\Big] \nonumber\\
&=  \int_{r_o}^{R_\mathrm{LoS}} \Big( 1 -  {{e^{ - t P_{\mathrm{UE}_o}^{\mathrm{mm}} \mathrm{M}_D \mathrm{M}_B L(r)}}}\Big) \nonumber\\
&\qquad\qquad \times \underbrace{\mathbb{E}\Big[ {{e^{ - tI_{\left(\widetilde{\lambda}_{\mathrm{mmUE}}\right)}^{\mathrm{mm}} }}} \Big]}_{\widetilde{I}_\mathrm{mm}(t)} \widetilde{f}_{\left| {{Y}} \right|} \left( r \right) dr,
\end{align}
where $P_{\mathrm{UE}_o}^{\mathrm{mm}}=\hbar_\mathrm{mm} L\left(r\right)$, and $\widetilde{f}_{\left| {{Y}} \right|} \left( \cdot \right)$ is the modified PDF of the distance $\left| {{Y}} \right|$ between a mmUE and its nearest mmWave BS under the constraint $\left| {{Y}} \right| \geq r_o$, which is
\begin{align}\label{new_PDF_mmWave}
\widetilde{f}_{\left| {{Y}} \right|} \left( r \right)=\frac{2\pi \lambda_\mathrm{mm} r}{e^{-\pi \lambda_\mathrm{mm} r_o^2}}\exp\left(-\pi \lambda_\mathrm{mm} r^2 \right).
\end{align}

We next calculate  $\widetilde{I}_\mathrm{mm}(t)$  as
\begin{align}\label{I_t_tilde}
&\widetilde{I}_\mathrm{mm}(t) \nonumber\\
&= \exp \left\{ { - 2\pi  \widetilde{\lambda}_{\mathrm{mmUE}}^{\mathrm{mm}} \int_{r_o}^{R_\mathrm{LoS}} {\left(1-\mathbb{E}\left[e^{-tP_{\mathrm{UE}_j}^{\mathrm{mm}} \widetilde{G}_j L\left( y \right)} \right]\right) y} dy} \right\}.
\end{align}
By using the {law} of total expectation, we can directly obtain
\begin{align}\label{averge_exp_I_t}
&\mathbb{E}\left[e^{-tP_{\mathrm{UE}_j}^{\mathrm{mm}} \widetilde{G}_j L\left( y \right)} \right] \nonumber\\
&=\sum\limits_{\ell  \in \left\{ {{\mathrm{M_B}},{\mathrm{m_B}}} \right\}} \sum\limits_{k \in \left\{ {{\mathrm{M_D}},{\mathrm{m_D}}} \right\}} \mathrm{Pr}_{\ell k} {\mathbb{E}_{P_{\mathrm{UE}_j}^{\mathrm{mm}}}\left[e^{-tP_{\mathrm{UE}_j}^{\mathrm{mm}} \ell k  L\left( y \right)}\right]} \nonumber\\
&\hspace{-0.3 cm}=\sum\limits_{\ell  \in \left\{ {{\mathrm{M_B}},{\mathrm{m_B}}} \right\}} \sum\limits_{k \in \left\{ {{\mathrm{M_D}},{\mathrm{m_D}}} \right\}} \mathrm{Pr}_{\ell k}  \int_{r_o}^{R_\mathrm{LoS}} e^{-t\hbar_\mathrm{mm} L\left(z\right) \ell k  L\left( y \right)} \widetilde{f}_{\left| {{Y}} \right|} \left( z \right) dz.
\end{align}
Substituting \eqref{new_PDF_mmWave} and \eqref{averge_exp_I_t} into \eqref{I_t_tilde}, after some manipulations, we obtain \eqref{mmWave_interfernce_average}, and complete the proof.

\section*{Appendix D: A detailed derivation of \eqref{NL_mmWave1}}
\label{App:theo_1}
\renewcommand{\theequation}{D.\arabic{equation}}
\setcounter{equation}{0}

 In the noise-limited scenario, the throughput is calculated as
 \begin{align}\label{NL_mmWave2}
& C_\mathrm{mm}^{\rm n}={\left(1-\tau\right)} \mathrm{BW}_\mathrm{mm} \mathbb{E}\left[\log_2\left(1+\frac{S_{\left(\lambda_\mathrm{mm}\right)}^{\mathrm{mm}}}{
{\sigma^2}}\right)\right].
 \end{align}
Considering the convexity of $\log_2(1+a e^x)$ for $a>0$ and  using Jensen's inequality,  the above can be lower bounded as
\begin{align}\label{NL_mmWave3}
& C_\mathrm{mm}^{\rm L}={\left(1-\tau\right)} \mathrm{BW}_\mathrm{mm} \log_2\left(1+\frac{e^{\zeta_3}}{
{\sigma^2}}\right),
\end{align}
where $\zeta_3=\mathbb{E}\left[\ln(S_{\left(\lambda_\mathrm{mm}\right)}^{\mathrm{mm}})\right]$. We can obtain $\zeta_3$ as
\begin{align}\label{zeta_3}
&\zeta_3=\mathbb{E}\left[\ln( P_{\mathrm{UE}_o}^{\mathrm{mm}} \mathrm{M}_D \mathrm{M}_B L\left({\left| {{Y_{{o}}}} \right|}\right)\mathbf{1}\left(\left| {{Y_{{o}}}} \right|< R_\mathrm{LoS}\right)) \right]\nonumber\\
&= \mathbb{E}\left[\ln(\hbar_\mathrm{mm}\mathrm{M}_D \mathrm{M}_B  L^2\left({\left| {{Y_{{o}}}} \right|}\right)\mathbf{1}\left(\left| {{Y_{{o}}}} \right|< R_\mathrm{LoS}\right)) \right]\nonumber\\
&\hspace{-0.2cm}=\underbrace{\int_{r_o}^{R_\mathrm{LoS}} (\ln(\hbar_\mathrm{mm}\mathrm{M}_D \mathrm{M}_B)+2(\ln\beta^{\rm{mm}}_\mathrm{LoS}-\alpha^{\rm{mm}}_\mathrm{LoS}\ln r)) \widetilde{f}_{\left| {{Y}} \right|} \left( r\right) dr}_{\widetilde{\varphi}_\mathrm{mm}(\lambda_\mathrm{mm})},
\end{align}
where $\widetilde{f}_{\left| {{Y}} \right|} \left( \cdot \right)$ is given by \eqref{new_PDF_mmWave}. From \eqref{zeta_3}, we see that ${\widetilde{\varphi}_\mathrm{mm}(\lambda_\mathrm{mm})}$ increases with $\lambda_\mathrm{mm}$ because when more BSs are deployed, a mmUE  becomes closer to its associated BS, which reduces the pathloss. After some manipulations, we can obtain ${\widetilde{\varphi}_\mathrm{mm}(\lambda_\mathrm{mm})}$ as \eqref{varphi_123} and complete the proof.

\bibliographystyle{IEEEtran}

\end{document}